\PassOptionsToPackage{unicode}{hyperref}
\PassOptionsToPackage{hyphens}{url}
\PassOptionsToPackage{dvipsnames,svgnames,x11names}{xcolor}
\documentclass[12pt]{article}
\usepackage{enumitem}
\usepackage{amsmath,amssymb}
\usepackage{bm}
\usepackage{colonequals}
\usepackage{iftex}
\usepackage{booktabs}
\usepackage{tabularx}
\usepackage{titlesec}
\ifPDFTeX
  \usepackage[T1]{fontenc}
  \usepackage[utf8]{inputenc}
  \usepackage{textcomp} % provide euro and other symbols
\else % if luatex or xetex
  \usepackage{unicode-math}
  \defaultfontfeatures{Scale=MatchLowercase}
  \defaultfontfeatures[\rmfamily]{Ligatures=TeX,Scale=1}
\fi
\usepackage{lmodern}
\ifPDFTeX\else  
\fi
\IfFileExists{upquote.sty}{\usepackage{upquote}}{}
\IfFileExists{microtype.sty}{% use microtype if available
  \usepackage[]{microtype}
  \UseMicrotypeSet[protrusion]{basicmath} % disable protrusion for tt fonts
}{}
\makeatletter
\@ifundefined{KOMAClassName}{% if non-KOMA class
  \IfFileExists{parskip.sty}{%
    \usepackage{parskip}
  }{% else
    \setlength{\parindent}{0pt}
    \setlength{\parskip}{6pt plus 2pt minus 1pt}}
}{% if KOMA class
  \KOMAoptions{parskip=half}}
\makeatother
\usepackage{xcolor}
\makeatletter
\ifx\paragraph\undefined\else
  \let\oldparagraph\paragraph
  \renewcommand{\paragraph}{
    \@ifstar
      \xxxParagraphStar
      \xxxParagraphNoStar
  }
  \newcommand{\xxxParagraphStar}[1]{\oldparagraph*{#1}\mbox{}}
  \newcommand{\xxxParagraphNoStar}[1]{\oldparagraph{#1}\mbox{}}
\fi
\ifx\subparagraph\undefined\else
  \let\oldsubparagraph\subparagraph
  \renewcommand{\subparagraph}{
    \@ifstar
      \xxxSubParagraphStar
      \xxxSubParagraphNoStar
  }
  \newcommand{\xxxSubParagraphStar}[1]{\oldsubparagraph*{#1}\mbox{}}
  \newcommand{\xxxSubParagraphNoStar}[1]{\oldsubparagraph{#1}\mbox{}}
\fi
\makeatother

\usepackage{longtable,booktabs,array}
\usepackage{calc} % for calculating minipage widths
\usepackage{etoolbox}
\makeatletter
\patchcmd\longtable{\par}{\if@noskipsec\mbox{}\fi\par}{}{}
\makeatother
\IfFileExists{footnotehyper.sty}{\usepackage{footnotehyper}}{\usepackage{footnote}}
\makesavenoteenv{longtable}
\usepackage{graphicx}
\makeatletter
\def\maxwidth{\ifdim\Gin@nat@width>\linewidth\linewidth\else\Gin@nat@width\fi}
\def\maxheight{\ifdim\Gin@nat@height>\textheight\textheight\else\Gin@nat@height\fi}
\makeatother
\setkeys{Gin}{width=\maxwidth,height=\maxheight,keepaspectratio}
\makeatletter
\def\fps@figure{htbp}
\makeatother

\makeatletter
\@ifpackageloaded{caption}{}{\usepackage{caption}}
\AtBeginDocument{%
\ifdefined\contentsname
  \renewcommand*\contentsname{Table of contents}
\else
  \newcommand\contentsname{Table of contents}
\fi
\ifdefined\listfigurename
  \renewcommand*\listfigurename{List of Figures}
\else
  \newcommand\listfigurename{List of Figures}
\fi
\ifdefined\listtablename
  \renewcommand*\listtablename{List of Tables}
\else
  \newcommand\listtablename{List of Tables}
\fi
\ifdefined\figurename
  \renewcommand*\figurename{Figure}
\else
  \newcommand\figurename{Figure}
\fi
\ifdefined\tablename
  \renewcommand*\tablename{Table}
\else
  \newcommand\tablename{Table}
\fi
}
\@ifpackageloaded{float}{}{\usepackage{float}}
\floatstyle{ruled}
\@ifundefined{c@chapter}{\newfloat{codelisting}{h}{lop}}{\newfloat{codelisting}{h}{lop}[chapter]}
\floatname{codelisting}{Listing}

\makeatother
\makeatletter
\@ifpackageloaded{caption}{}{\usepackage{caption}}
\@ifpackageloaded{subcaption}{}{\usepackage{subcaption}}
\makeatother

\RequirePackage{amsthm,amsmath,amsfonts,amssymb}
\RequirePackage{hyperref}
\RequirePackage{cleveref}

\newtheoremstyle{inferenceaimstyle}
  {\topsep}   % Space above
  {\topsep}   % Space below
  {\itshape}  % Body font
  {}          % Indent amount
  {\bfseries} % Theorem head font
  {}         % Punctuation after theorem head
  { }         % Space after theorem head
  {\thmname{#1}\thmnumber{ #2}:} % Theorem head spec

\theoremstyle{inferenceaimstyle}

\usepackage{thmtools}
\crefname{lemma}{Lemma}{Lemmas}
\declaretheorem[name=Lemma]{lemma}

\crefname{proposition}{Proposition}{Propositions}

\crefname{theorem}{Theorem}{Theorems}
\declaretheorem[name=Theorem]{theorem}

\crefname{figure}{Figure}{Figures}
\crefname{table}{Table}{Tables}
\crefname{section}{Section}{Sections}
\crefname{equation}{Equation}{Equations}

\ifLuaTeX
  \usepackage{selnolig}  % disable illegal ligatures
\fi
\usepackage[]{natbib}
\usepackage{bookmark}
\RequirePackage{algorithm}
\RequirePackage{algorithmic}
\usepackage{subcaption}

\IfFileExists{xurl.sty}{\usepackage{xurl}}{} % add URL line breaks if available
\hypersetup{
  pdftitle={Title},
  pdfauthor={Author 1; Author 2},
  pdfkeywords={3 to 6 keywords, that do not appear in the title},
  hidelinks=true,
  linkcolor={blue},
  filecolor={Maroon},
  citecolor={Blue},
  urlcolor={Blue},
  pdfcreator={LaTeX via pandoc}}

\definecolor{changed_color}{rgb}{0,0,0}

\newcommand{\anon}{1}

\begin{document}

\def\spacingset#1{\renewcommand{\baselinestretch}%
{#1}\small\normalsize} \spacingset{1}

%%%%%%%%%%%%%%%%%%%%%%%%%%%%%%%%%%%%%%%%%%%%%%%%%%%%%%%%%%%%%%%%%%%%%%%%%%%%%%

\if1\anon
{
  \title{\bf {Empirical Bayes Shrinkage of Functional Effects, with Application to Analysis
  of Dynamic eQTLs}}
% Empirical Bayes Shrinkage for Functional
% Effects with Application to Dynamic eQTL
% Analysis
  \author{Ziang Zhang \hspace{.2cm}\\
    Department of Human Genetics, University of Chicago, Chicago, IL \\
    \\
    Peter Carbonetto\\
    Department of Human Genetics, University of Chicago, Chicago, IL \\
    \\
    % AUTHORS NAME:\\
    % AFFILIATIONS:\\
    % \\
    Matthew Stephens \\
    Departments of Statistics and Human Genetics, \\
    University of Chicago, Chicago, IL
    }
  \maketitle
} \fi

\if0\anon
{
  \bigskip
  \bigskip
  \bigskip
  \begin{center}
    {\LARGE\bf Title}
\end{center}
  \medskip
} \fi

\bigskip
% Should contain an unstructured abstract of 200 words.
\begin{abstract}
We introduce functional adaptive shrinkage (FASH), an empirical Bayes
method for joint analysis of observation units in which each unit
estimates an {\em effect function} at several values of a continuous
condition variable. The ideas in this paper are motivated by dynamic
expression quantitative trait locus (eQTL) studies, which aim to
characterize how genetic effects on gene expression vary with time or
another continuous condition. 
FASH integrates a broad family of
Gaussian processes defined through linear differential operators into
an empirical Bayes shrinkage framework, enabling adaptive smoothing and borrowing of information across units. 
This provides improved estimation of effect functions and principled hypothesis testing, allowing straightforward computation of significance measures such as local false discovery and false sign rates.
{\color{changed_color} To encourage conservative inferences, we propose a simple prior-adjustment method that has theoretical guarantees and can be more broadly used with other empirical Bayes methods.}
%
% we develop a Bayes-factor–based adjustment to the FASH prior that
% can be incorporated into general empirical Bayes shrinkage
% procedures.
%
We illustrate the benefits of FASH by reanalyzing dynamic
eQTL data on cardiomyocyte differentiation from induced pluripotent
stem cells. FASH identified novel dynamic eQTLs, revealed diverse
temporal effect patterns, and provided improved power compared with
the original analysis. More broadly, FASH offers a flexible
statistical framework for joint analysis of functional data, with
applications extending beyond genomics. To facilitate use of FASH in
dynamic eQTL studies and other settings, we provide an accompanying R
package at \url{https://github.com/stephenslab/fashr}.
\end{abstract}

% From https://www.tandfonline.com/action/authorSubmission?show=instructions&journalCode=uasa20:
% Should contain between 3 and 5 (or 6) keywords.
\noindent%

% % I've changed the spacing temporarily just to make it easier to read. -Peter
% \spacingset{1} 
\spacingset{1.8} % DON'T change the spacing!
\captionsetup{font={stretch=1.0}}

\section{Introduction}
\label{sec-intro}

Perturbing a system, and measuring the changes that result, is a
classic scientific way to understand a
system \citep{derisi1997exploring, heller1997discovery,
hughes2000functional}. Modern genomics experiments therefore often
measure the behaviors of many genomic units under a variety of
perturbations or treatments. Some of these experiments involve
measuring how behaviors change as a function of some underlying
continuous variable. 
%
% \citep{gtex2020gtex, knowles2017allele} 
%
For example, recent expression quantitative trait locus (eQTL) and
allele-specific expression (ASE) studies measure how the effects of
genetic variants on gene expression change over time, or with respect
to a continuous condition such as oxygen level
\citep{cuomo2020single, elorbany2022single, francesconi2014effects, 
Gutierrez-Arcelus2020, kang2023mapping, nathan2022single,
soskic2022immune, strober2019dynamic}. These studies, which we refer
to as ``dynamic eQTL studies'', typically produce noisy measurements
of eQTL effects for genomic units (gene-variant pairs) at multiple
settings of a continuous condition.
%
% ---that is, effect estimates together with their standard
% errors---
%
The goals of dynamic eQTL studies include estimating the underlying
eQTL effect functions (e.g., to characterize how the eQTL effect
changes over time) and identifying which eQTL effects deviate from
some ``null'' or ``baseline'' behavior. This null or baseline could be
defined as no change over time, or a linear change over time.
%
% {\color{changed_color}Dynamic eQTL analyses build on the classic idea
% of using perturbations to probe system behavior.  By introducing
% temporal or environmental changes as perturbing variables, these
% studies can uncover transient and context-specific genetic effects
% that are invisible in steady-state data, thereby revealing how genetic
% variation influences gene regulation across dynamic cellular states.}
%

Dynamic eQTL studies involve massive data, generating noisy
measurements from thousands or even millions of genomic
units. Therefore, statistical inferences can be greatly improved by
sharing information across units rather than treating each unit
separately as is typically done.
%
% \citep{ferguson2012new, flutre2013statistical}. 
%
Empirical Bayes (EB) \citep{efron2008microarrays, efron2009empirical}
provides a powerful framework for accomplishing this. EB methods learn
a common prior distribution from all units, which is usually centered
on the null or baseline, and then it improves the effect estimates for
each unit by shrinking them toward this null. An attractive feature of
EB is that the amount of shrinkage is adaptive in that it depends on
both the prior distribution (for example, experiments with more null
units result in prior distributions that shrink more) and on the
standard errors of the effect estimates (the noisier the estimate, the
stronger the shrinkage) \citep{ash}.
%
% In addition, when large-scale hypothesis testing is required, the EB
% framework provides a natural way to compute quantities such as the
% false discovery rate.
%
EB methods and convenient software implementations are widely available and 
are becoming more frequently used for studies involving one
condition \citep{ash, willwerscheid2025ebnm} or several
%
%  (categorical, nominal) 
%
conditions \citep{mash, liu2024flexible}.

{\color{changed_color}
In this work, we develop new empirical Bayes methods and software for
estimating {\em effect functions} that are defined on a
continuously-valued space such as time.
%
% Here we develop a new EB method and software for data involving
% {\em continuously-valued conditions.}
%
To model the effect functions, we use Gaussian processes
(GP) \citep{mackay}, probably the simplest and most widely used
approach to define a prior distribution for functions. We exploit a
family of GPs which we call the ``$L$-GP family'' that can encode
diverse functional shapes and behaviors. Members of this family are
defined by a single scalar parameter, $\sigma$, that quantifies
deviations from a null/baseline model space (e.g., the space of
constant functions or the space of linear
functions) \citep{lindgren2008second, yue2014bayesian, zhang2024model,
zhang2025efficient}. We show how the $L$-GP family can be combined
with adaptive shrinkage ideas \citep{ash} to {\em adaptively shrink
functions.} We call the new methods
``\textbf{F}unctional \textbf{A}daptive \textbf{SH}rinkage'' (FASH).}

% FASH uses a prior that is a mixture of $L$-GP distributions over a
% dense grid of $\sigma$ values.
% 
% Following the typical EB approach, 
%
% FASH then estimates these mixture weights by maximum-likelihood, and
% computes the posterior mean for each effect function given the
% estimated mixture weights. Therefore, FASH {\em adaptively shrinks}
% the effect estimates towards the baseline model; for example, if the
% data are largely consistent with the baseline model, the prior
% distribution will apply more weight to small settings of $\sigma$,
% resulting in stronger shrinkage toward the baseline.

{\color{changed_color}
% In addition to providing shrinkage-based estimates of the effect
% functions, 
% 
The combination of these ideas also provides a flexible way to perform
hypothesis testing on effect functions; that is, to assess whether an
effect function significantly departs from a given
%
% (composite) 
%
baseline model. Hypothesis testing on effect functions through FASH is very
different from---and has important benefits over---the {\em ad hoc}
approaches that have been used in dynamic eQTL
studies \citep{cuomo2020single, elorbany2022single,
francesconi2014effects, Gutierrez-Arcelus2020, kang2023mapping,
nathan2022single, soskic2022immune, strober2019dynamic}. For
example, \cite{soskic2022immune} performed hypothesis testing by
comparing models of gene expression with and without a
$\mbox{genotype} \times \mbox{time}$ interaction term. The problem
with this approach is that it fails to account for other potentially
interesting time-dependent effects in which the changes over time are
nonlinear. (Recognizing this limitation, \cite{soskic2022immune} also
compared models with and without a quadratic interaction term,
$\mbox{genotype} \times \mbox{time}^2$, but this still makes a
restrictive assumption about how the eQTL effects vary with respect to
time.) By contrast, the hypothesis testing in FASH does not make such
assumptions (see \cref{subsec-toy} for an illustration). As a result, our approach is more flexible and powerful.}  {\color{changed_color}
Furthermore, FASH has another benefit over the approaches that have
been used in dynamic eQTL studies: previous approaches involve fitting
models using the individual-level data, whereas FASH operates on
summary statistics, which are sometimes available when the
individual-level data are not (e.g., due to HIPAA requirements or
other data privacy concerns). Therefore, FASH can also be applied to
settings where only summary statistics are available.}

{\color{changed_color} Finally, we note that our new methods also address a
well known but often unaddressed limitation of empirical Bayes
methods.}  {\color{changed_color} EB methods
generally require that the prior distribution accurately captures both
the null and alternative hypotheses to achieve calibrated inference.
However, accurately specifying the prior under the alternative is
often unrealistic in practice, and the resulting inferences may
therefore lose calibration or become anticonservative.}  To help
ensure conservative inferences, we develop a Bayes factor (BF)-based
adjustment to the 
%
% estimated 
%
prior. While we only implemented this idea for FASH in this paper, we
note that it is a general technique that could potentially be
used for other EB-based shrinkage methods.

All the methods described in this paper have
been implemented in an R package, \texttt{fashr} (``functional adaptive
shrinkage in R''), which is available
at \url{https://github.com/stephenslab/fashr}. To our knowledge, the
previous dynamic eQTL studies did not disseminate dedicated software
tools for their hypothesis testing pipelines, so \texttt{fashr} represents the
first integrated software for dynamic eQTL studies. Further, the FASH
method and software implementation are quite general, and therefore we
forsee the potential to apply FASH to other data sets beyond dynamic
eQTL studies. The R package has detailed documentation and an
accessible interface, as well as a vignette illustrating recommended
usage of \texttt{fashr} for a dynamic eQTL data set. 

\subsection{Organization of Paper}

The rest of the paper is organized as follows.
In \cref{subsec-toy}, we present a motivating example illustrating shrinkage toward a baseline model and adaptive shrinkage via borrowing information across observation units.
In \cref{sec-method}, we introduce our FASH approach together with a BF-based adjustment for conservative hypothesis testing.
In \cref{sec-application}, we apply FASH to reanalyze the dynamic eQTL dataset of \citet{strober2019dynamic}, spanning 16 days of cardiomyocyte differentiation from induced pluripotent stem cells and comprising over one million gene--variant pairs.
This analysis identifies novel dynamic eQTLs, reveals diverse temporal effect patterns, and provides a more accurate characterization of effect functions than the parametric interaction analysis in the original study, which imposes more rigid assumptions.
Finally, \cref{discussion} summarizes our contributions, outlines current limitations of our approach, and discusses potential extensions to other areas.

% {\color{purple}
% 1. datasets with many effects measured under many continuously labeled conditions.
% 2. example: dynamic eQTL studies.
% 3. highlevel idea of functional data analysis.
% 4. want to infer the unknown function, and test hypothesis.
% 5. shrinkage is fruitful for such many datasets, many condition setting
% 6. existing method does not work for functional setting, where the continuous label provides correlation information.
% 7. the contribution of this work, fash, as well as a further BF-based FDR correction.
% 8. application of the method to re-analyze the dynamic eQTL example.}

\section{A Motivating Example}
\label{subsec-toy}

\begin{figure}[!t]
\centering
\includegraphics[width=1.0\textwidth]{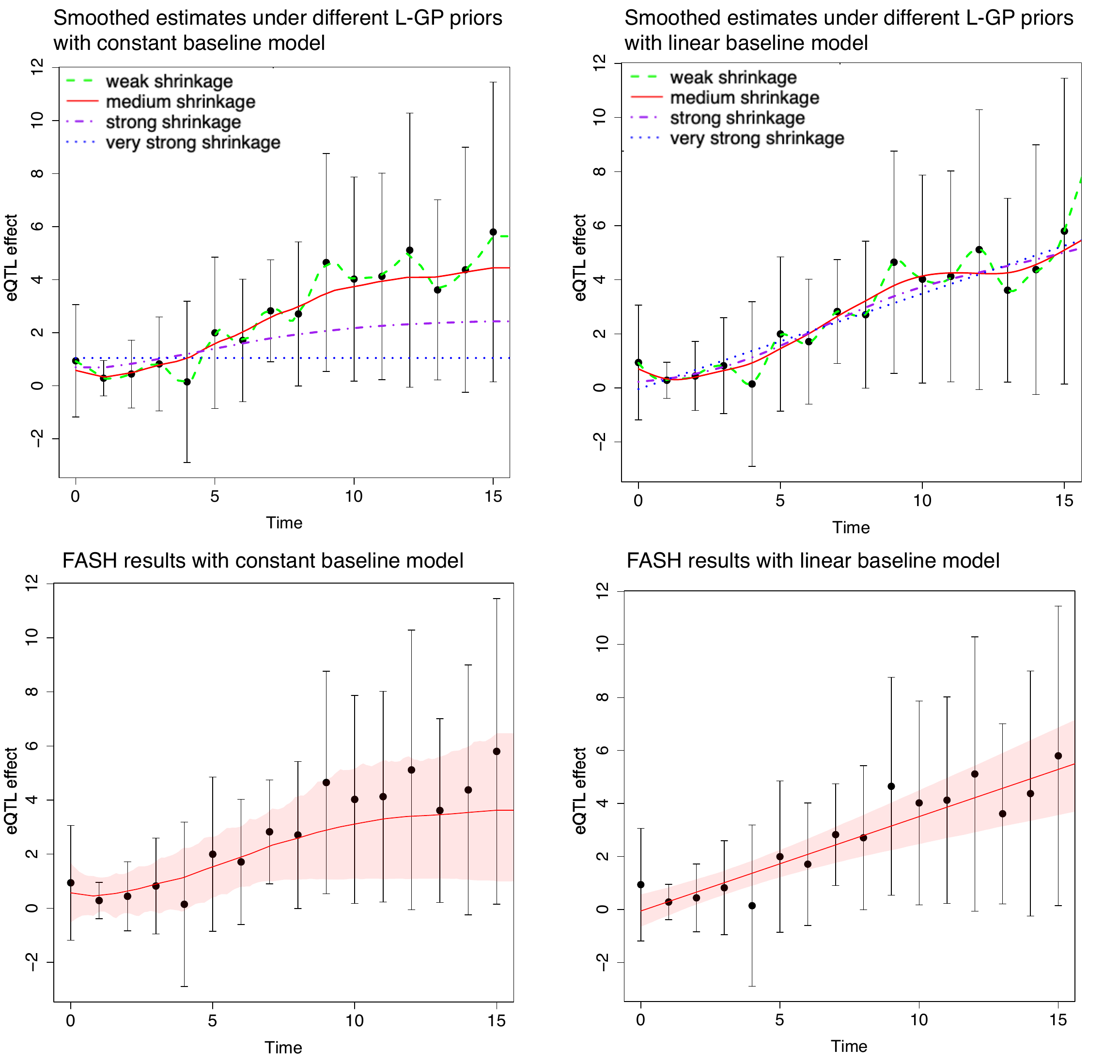}
\caption{Example illustrating the use of FASH to analyze a 
dynamic eQTL data set. In all panels, the original eQTL effect
estimates are shown as black dots, and vertical error bars depict $\pm
2$ standard errors. The top two panels show the smoothed estimates
defined as the posterior means from the $L$-GP method; the different
estimates are obtained by different levels of shrinkage toward the
constant and linear baseline models (left and right panels,
respectively). The bottom two panels summarize the results from the
proposed FASH method, with constant and linear baseline models (left
and right panels, respectively). The constant and linear baseline
models are defined using $L$-GP processes with $L = D^1$ and $L =
D^2$, respectively (see \cref{subsec-ASH}). The posterior mean effect
function is shown as a solid red line; the shaded region shows the
95\% credible interval.}
\label{fig:toy}
\end{figure}

To motivate the key ideas of our approach, we first give an example
from the dynamic eQTL application that is presented in more detail
below. For this example, we focus on data and results for a single
gene-variant pair: gene {\em FZD6} and SNP rs28392906. The black dots
in each of the panels in \cref{fig:toy} show the eQTL effect estimates at the different time points for this
gene-variant pair. The standard errors of the eQTL effects are shown
as vertical error bars.
%
% $\hat{\beta}_j(t_{j,r})$ 
%

Part of the statistical analysis involves improving the accuracy of
these ``noisy'' eQTL estimates by shrinking them toward a baseline
model. The baseline models in FASH assume that the true effect
function under this model varies ``smoothly'' over time. But the exact
baseline model we choose is guided by the scientific question of
interest. Here we consider two questions: (i) Is the eQTL effect {\em
dynamic}---that is, does it change over time?  (ii) Does the eQTL
effect change over time in a way that deviates from a linear
trend---that is, does it exhibit nonlinear changes over time?  The
null hypotheses corresponding to these baseline models are: (i) the
effect function is constant; (ii) the effect function is linear. Our
methods can shrink the effect functions towards either of these
baseline models, as well as other baseline models.

First, to illustrate this smoothing of the effect estimates,
%
% Given a baseline model, a typical smoothing technique will have a
% parameter that controls the strength of the shrinkage toward the
% baseline.  
%
we show the result of applying the $L$-GP model with different levels
of shrinkage, and with either the constant (\cref{fig:toy}, top-left
panel) or linear baseline (top-right panel).
From these plots, we see that the smoothed estimates vary greatly depending on the shrinkage level.
At the strongest shrinkage level, the smoothed estimates reduce to a constant function that averages across time points or a straight line from weighted linear regression. 
At the weakest shrinkage level, the smoothed estimates become a curve that interpolates all data points exactly.
%
% Depending on the baseline model, the smoothed
% estimates range from interpolating all observed effects to simply
% averaging across time points or fitting a weighted linear
% regression.
% 
Thus, determining the appropriate level of shrinkage is critical to
the analysis.

The bottom two panels in \cref{fig:toy} visualize the posterior
distributions of the effect functions obtained by applying FASH to
these data. FASH uses $L$-GP priors, but adapts these priors to the
data by pooling the information from all the available gene-variant
pairs.
% 
% (Only one of the gene-variant pairs is shown in \cref{fig:toy}.)  
%
Additionally, the shrinkage level is adapted separately for each
gene-variant pair depending on how closely the data match the baseline
model; data close to the baseline model tend to be shrunk more
strongly than data that are far away. (Different shrinkage patterns
for other gene-variant pairs are shown in
\cref{sec-application}). This is sometimes referred to
as ``global-local'' shrinkage \citep{polson2010shrink}.

{\color{changed_color} In this example, the data are not consistent
with a baseline model in which there is no change in the effect over
time. Therefore, the estimates are shrunk weakly with respect to this
baseline (\cref{fig:toy}, bottom-left). The hypothesis test from this
baseline model is a test of whether the effects remain unchanged over
time; clearly, there is strong evidence for rejecting this
hypothesis. (Hypothesis testing in FASH is introduced more formally
below.) Note that this test does not make any specific assumptions
about {\em how} the effects change over time (unlike analyses in many
published dynamic eQTL studies which do make such assumptions).}

{\color{changed_color} In the bottom-right panel of \cref{fig:toy}, we
see that the data are much more consistent with a linear change in effect over time. Therefore, the effect estimates are strongly
shrunk toward the linear baseline model. Since these data agree well
with the linear baseline model, it follows that there is little evidence
for rejecting the hypothesis of a linear change in eQTL effects, and therefore we view a roughly linear increase in the eQTL effects over time as a plausible description of these data. 
The two tests address different questions: the first assesses whether the effects change over time at all, whereas the second assesses whether there is evidence that those changes are nonlinear.
We give other examples of hypothesis testing in \cref{sec-application}.}

% In all cases, the amount of shrinkage also adapts to the standard
% errors of the observations.
%
% the posterior mean tends to be close to high-precision observations
% while smoothing over noisier observations.
%
% FASH identifies this eQTL as dynamic but not nonlinear, based on the
% corresponding false discovery rates outlined
% in \cref{subsec-inference}.

% In the following sections, we describe the main procedures of the
% FASH approach, emphasizing how it is designed to address these
% inferential questions. Further computational details of the fitting
% algorithm are provided in Appendix D.

\section{Functional Adaptive Shrinkage}
\label{sec-method}

\subsection{Notation}
\label{subsec-notation}

% In this paper, we adopt the following mathematical notation unless explicitly stated otherwise.
% We assume vectors and matrices are respectively shown as bold lower case ($\boldsymbol{a}$) and upper case ($\mathbf{A}$) letter.
% Scalars are shown as plain, lower case letter ($a$).
% The notation $[J] = \{1,...,J\}$ denotes the enumeration.
% The notation $|A|$ denotes the cardinality of the set $A$; when used with a matrix $\mathbf{A}$, then $|\mathbf{A}|$ refers to the determinant.
% A normal distribution is denoted by its mean and variance as $N(\mu,\sigma^2)$, with density denoted as $d\mathcal{N}(.;\mu,\sigma^2)$.
% We use $iid$ to denote independently and identically distributed; and $ind$ to denote independently distributed.

We use the following conventions in the mathematical expressions. For
an integer $J$, we write $[J] = \{1,\ldots,J\}$ for the index set.
For a set $A$, $|A|$ denotes its cardinality. The normal distribution
with mean $\mu$ and variance $\sigma^2$ is written as
$N(\mu, \sigma^2)$, with density
$d\mathcal{N}(\,\cdot\,; \mu, \sigma^2)$.  We abbreviate ``independent
and identically distributed'' as ``iid'', and we abbreviate
``independent'' as ``ind''.  We use $\mathbb{Z}^{+}$ to denote the set
of positive integers, $\mathbb{R}$ for the set of real-valued vectors,
$\mathbb{R}^p$ for the set of real-valued vectors of length $p$,
$\mathbb{R}^{p\times q}$ for the set of real-valued matrices of size
$p \times q$, and $C^p(\Omega)$ for the set of functions
$\Omega \rightarrow \mathbb{R}$ that are $p$-times continuously
differentiable.

\subsection{Problem Setup}
\label{sec:data}

Let $J \in \mathbb{Z}^{+}$ be the number of {observation units}.  For
each unit $j \in [J]$, we assume that we have effect estimates at
several values of a continuous condition variable $t \in \mathbb{R}$.
We denote these effect estimates by $\hat{\beta}_j(t_{j1}), \ldots,
\hat{\beta}_j(t_{jR_j}) \in \mathbb{R}$, where 
$R_j$ is the number of settings of $t$ where we have effect estimates
for unit $j$. {\em Our main aim is to estimate an underlying effect
function $\beta_j$, a mapping from $t \in \Omega$ to $\mathbb{R}$.} We
are particularly interested in settings where the number of units,
$J$, is large---say, hundreds or thousands, or perhaps more---so that
we can pool information to improve estimation of the underlying effect
functions. It is also important that have observations at several
values of $t$ for each $j \in [J]$.

In a dynamic eQTL study, $J$ is the number of gene-variant pairs, and
each $\beta_j(t_{jr})$ is the estimated effect of a genetic variant on
gene expression at time point $t_{jr}$ (or setting $t_{jr}$ of the
continous variable). When there is only one genetic variant associated
with each gene, $J$ is simply be the number of genes.  In our dynamic
eQTL case study below (\cref{sec-application}), $J$ is over 1 million,
and $R_j = 16$ for all $j \in [J]$. Note that in some dynamic eQTL
studies, expression is only measured at 2 or 3 time points, which is
probably too few for our methods to be useful.

\subsection{Empirical Bayes for Functional Data Analysis}
\label{subsec-model}

% We consider the following problem. 
% Assume there exists $J \in \mathbb{Z}^+$ datasets, for each dataset $j$, we measure the effects under $R_j \in \mathbb{Z}^+$ levels of a continuous condition variable ($t_{j1},...,t_{jR_j}$).
% The true effect of the $j$th dataset under condition $t_{j,r}$ is denoted as $\beta_j(t_{j,r})$, and observed effect estimate is denoted as $\hat{\beta}_j(t_{j,r})$, and its standard error is denoted as $s_{j,r}$.
% \begin{equation}\label{fash_likelihood}
%     \begin{aligned}
%         \boldsymbol{\hat{\beta}_j} &= [\hat{\beta}_j(t_{j,1}), \ldots, \hat{\beta}_j(t_{j,R_j})]^T, \\
%         \boldsymbol{{s}}_j &= [{s}_{j,1}, \ldots, {s}_{j,R_j}]^T, \quad j \in [J],\\
%         \hat{\beta}_j(t_{j,r})\ &|\ \beta_j(t_{j,r}) \overset{\text{iid}}{\sim} N\left(\beta_j(t_{j,r}), {s}_{j,r}^2\right), \quad r \in [R_j].
%     \end{aligned}
% \end{equation}

Our methods assume that the effect estimates are independent and
normally distributed given the underlying effect function:
\begin{equation}
\label{fash_likelihood}
\hat{\beta}_j(t_{jr}) \mid \beta_j, s_{jr}
\overset{\text{ind}}{\sim}\; N(\beta_j(t_{jr}), s_{jr}^2), 
\quad r \in [R_j],
\end{equation}
in which $s_{jr}$ denotes the standard error of the effect
estimate $\beta_j(t_{jr})$.
%
% Using the summary statistics \(\boldsymbol{\hat{\beta}} = (\boldsymbol{\hat{\beta}}_j)_{j=1}^J\) and 
% \( \boldsymbol{{s}} = (\boldsymbol{{s}}_j)_{j=1}^J \) as input, the inferential target is the unknown effect function $\beta_j$ for each $j$.
% In particular, it includes the following two questions:
% \begin{enumerate}
%     \item \textbf{Smoothing:} Characterize and visualize the inferred function $\beta_j$, or its functionals (such as derivatives) at observed or unobserved values $t$.
%     \item \textbf{Hypothesis Testing:} Among all the unknown functions $\{\beta_j\}_{j=1}^J$, test how many of the functions satisfy certain criterion $H_0:\beta_j \in S_0$ where $S_0$ denotes a certain space of functions, under false discovery or false sign rate control.
% \end{enumerate}
%
We further assume each underlying effect functions 
${\beta}_j$ are iid draws from a prior distribution 
for functions on
$\Omega \rightarrow \mathbb{R}$, denoted by $g_{\beta}$:
\begin{equation}
\label{fash_prior}
\beta_j \;\overset{\text{iid}}{\sim}\; g_{\beta}.
% \in \mathcal{G}_\beta.
\end{equation}

We have two main inference goals:
\begin{enumerate}

\item \textbf{Smoothing:} Recover and visualize the underlying 
effect function, ${\beta}_j$, or its functionals (e.g., derivatives),
at observed or unobserved values of $t$.

\item \textbf{Hypothesis testing:} Identify the units $j$ that deviate 
from a null hypothesis $H_0 : {\beta}_j \in S_0$, where $S_0$ denotes
some prespecified class of functions. In some applications, we might
also be interested in testing hypotheses of the form $H_0
: \mathcal{F}({\beta}_j) = 0$ for some functional $\mathcal{F}:
C^p(\Omega) \to \mathbb{R}$.  {For example, we might be interested in
testing whether the maximum of the function exceeds a given threshold
$\alpha$, which would be achieved with $\mathcal{F}({\beta}_j)
= \max_{t \,\in\, \Omega} {\beta}_j(t) - \alpha$.}
%
% {based on} false discovery or false sign rates.
%

\end{enumerate}
The first goal (``smoothing'') will be accomplished by computing the
posterior distribution of each effect function ${\beta}_j$:
\begin{equation}
\label{fash_posterior}
p({\beta}_j \mid \hat{\bm\beta}_j, \boldsymbol{s}_j, g_{\beta}) 
\propto p(\hat{\bm\beta}_j \mid {\beta}_j, \boldsymbol{s}_j) \,
g_{\beta}({\beta}_j),
% {\color{changed_color}{\hat{p}(\boldsymbol{\beta}) = 
% \prod_{j=1}^J \hat{g}_{\beta}(\boldsymbol{\beta}_j)}.}
\end{equation}
such that $\hat{\bm\beta}_j \colonequals
(\hat{\beta}_j(t_{j1}), \ldots, \hat{\beta}_j(t_{jR_j}))^T$,
$\boldsymbol{s}_j \colonequals (s_{j1}, \ldots, s_{jR_j})^T$ denote
the available data for unit $j \in [J]$. A smoothed point estimate of
${\beta}_j$ can be then obtained by its posterior mean. The second
goal (``hypothesis testing'') is accomplished by computing local false
discovery rates and local false sign rates (see below).
%
% defined later in \cref{equ:lfdr_def} and \cref{equ:lfsr}).
%

%
% We use $\hat{\bm\beta} \colonequals
% (\hat{\bm\beta}_j)_{j=1}^J$ and $\boldsymbol{s} \colonequals
% (\boldsymbol{s}_j)_{j=1}^J$ to denote the data, 
%

Specifying a single prior $g_{\beta}$ that works well across all data
sets is generally unrealistic, so it is much better if there is a
mechanism to learn or adapt a prior automatically based on the data.
We take an empirical Bayes approach to learning a prior: given a
predefined family of prior distributions, $\mathcal{G}_{\beta}$ (which
we will define below), we choose a $g_{\beta} \in \mathcal{G}_{\beta}$
that maximizes the log-likelihood of all the data,
\begin{equation}
\hat{g}_{\beta}
= \underset{g_{\beta} \,\in\, \mathcal{G}_{\beta}}{\mathrm{argmax}}
\sum_{j=1}^J \log \textstyle \int
\prod_{r = 1}^{R_j} p(\hat{\beta}_j(t_{jr}) \mid \beta_j, s_{jr})
\times g_{\beta}(\beta_j)
\, d\beta_j,
\label{eq:mle}
\end{equation}
and then all subsequent inferences use this estimated prior
$\hat{g}_{\beta}$.

Although our main aim is to make inferences about the unknown effect
functions, which involves reasoning about probability distributions on
functions, the underlying statistical computations are
straightforward in practice, reducing to probability distributions on
finite-dimensional spaces. Letting ${\bm\beta}_j \colonequals
(\beta_j(t_{j1}), \ldots, \beta_j(t_{jR_j}))$ denote the underlying
effect function at the values of $t$ where the effect estimates are
available, and ${\bm\beta} \colonequals \{{\bm\beta}_1, \ldots,
{\bm\beta}_J\}$, the modeling assumptions above imply the existence of
a prior $p({\bm\beta}; g_{\beta})$ such that
\begin{equation}
\label{fash_prior_fdd}
p(\boldsymbol{\beta}; {g_{\beta}}) =
\prod_{j=1}^J p(\boldsymbol{\beta}_j; {g_{\beta}}),
\end{equation}
where $p(\boldsymbol{\beta}_j; {g_{\beta}})$ denotes the
finite-dimensional distribution on $\boldsymbol{\beta}_j$ induced by
the prior $g_{\beta}$. For example, the statistical computations
needed for the maximum-likelihood estimation of the
prior \eqref{eq:mle} reduce to finite-dimensional integrals:
\begin{equation}
\hat{g}_{\beta} 
% &= 
% \underset{g_{\beta} \,\in\, \mathcal{G}_{\beta}}{\mathrm{argmax}} \;
% l(g_\beta) 
% \nonumber \\
% &= \underset{g_{\beta} \,\in\, \mathcal{G}_{\beta}}{\mathrm{argmax}}
% \sum_{j=1}^J \log p(\hat{\bm\beta}_j \mid \boldsymbol{s}_j)
% \nonumber \\
= \underset{g_{\beta} \,\in\, \mathcal{G}_{\beta}}{\mathrm{argmax}}
\sum_{j=1}^J \log \textstyle \int
p(\hat{\bm\beta}_j \mid \boldsymbol{\beta}_j, \boldsymbol{s}_j)
\, p(\boldsymbol{\beta}_j; g_{\beta}) \, d\boldsymbol{\beta}_j,
\end{equation}
where 
\begin{equation}
p(\hat{\bm\beta}_j \mid \boldsymbol{\beta}_j, \boldsymbol{s}_j) =
\prod_{r=1}^{R_j} d\mathcal{N}(\hat{\beta}_j(t_{jr}); 
\beta_j(t_{jr}), s_{jr}^2).
\end{equation}

\subsection{The Functional Adaptive Shrinkage Family of Priors}
\label{subsec-ASH}

% In this section, we consider the choice of prior that targets the above two goals.
% We start with defining the notion of \textit{baseline model} $S_0$, which is a space of simple function that is viewed as perfectly \textit{smoothed}.

We now describe the ``functional adaptive shrinkage'' family of prior
distributions on effect functions that addresses the goals outlined
above. The construction of this family begins with the specification
of a ``baseline model,'' $S_0$. This is done by specifying a
$p$th-order linear differential operator $L = \sum_{i=0}^p c_i D^i$,
where $D^i$ denotes the $i$th derivative operator and the
$c_i \in \mathbb{R}$ are specified coefficients. Then we define the
baseline model as
\begin{equation}
S_0 = \mathrm{Null}\{L\} = 
\{ \beta \in C^p(\Omega): L\beta = 0 \}.  
\end{equation}
For example, if $L = D^1$, then $S_0 = \text{span}\{1\}$ corresponds
to the class of constant functions; if $L = D^2$, then $S_0
= \text{span}\{1,t\}$ corresponds to the class of linear functions.
{More generally, for $L = D^p$, the baseline model corresponds to
the space of polynomials of order $p-1$.}

% Assume that each function $\beta_j$ is $p$ times continuously differentiable function with a compact support $\Omega$ (i.e. $\beta_i \in C^p(\Omega)$), and the baseline model could be specified through a $p$th order linear differential operator $L = \sum_{i=0}^pc_iD^i$ such that $\{c_0,...,c_p\}$ are known constants and $D^i$ denotes the $i$th order differentiation operator, such that:
% \begin{equation}
%     \begin{aligned}
%         S_0 &= \text{Null}\{L\}
%         = \{f\in C^p(\Omega): Lf = 0\}.
%     \end{aligned}
% \end{equation}

% For example, if we are interested in detecting when $\beta_j$ is non-constant function, then the baseline model corresponds to $S_0 = \text{span}\{1\}$, corresponding to $L = D$.
% If we are interested in testing whether any $\beta_j$ is nonlinear, then the baseline model becomes $S_0 = \text{span}\{1,t\}$.

Given $L$, we define an $L$-Gaussian process (``$L$-GP'') as a
Gaussian process satisfying
\begin{equation}
\label{equ:LGP}
L\beta = \sigma \xi,
\end{equation}
where $\xi$ is standard Gaussian white noise, and $\sigma>0$ is a
parameter that governs how much $\beta$ deviates from the baseline
model $S_0$. For conciseness, we write {this
as}
\begin{equation*}
\beta \sim L\text{-GP}(\beta; \sigma).
\end{equation*}
%
% with $\sigma$ controlling the magnitude of deviation allowed from
% the baseline model $S_0$.
%
When $L = D^p$, the $L$-GP corresponds to the $p$th-order Integrated
Wiener Process ($\text{IWP}_p$) \citep{shepp1966radon}.  The IWP prior
is closely connected to smoothing splines as its posterior mean
coincides with the spline estimator \citep{wahbaimproper}.  In what
follows, we focus on $L$-GPs of the $\text{IWP}_p$ form, while noting
that the method also applies to other types of $L$-GPs;
see \cite{lindgren2008second, yue2014bayesian, zhang2024model,
zhang2025efficient} for further background and examples.

\begin{figure}[!t]
\centering
\includegraphics[width=1.0\textwidth]{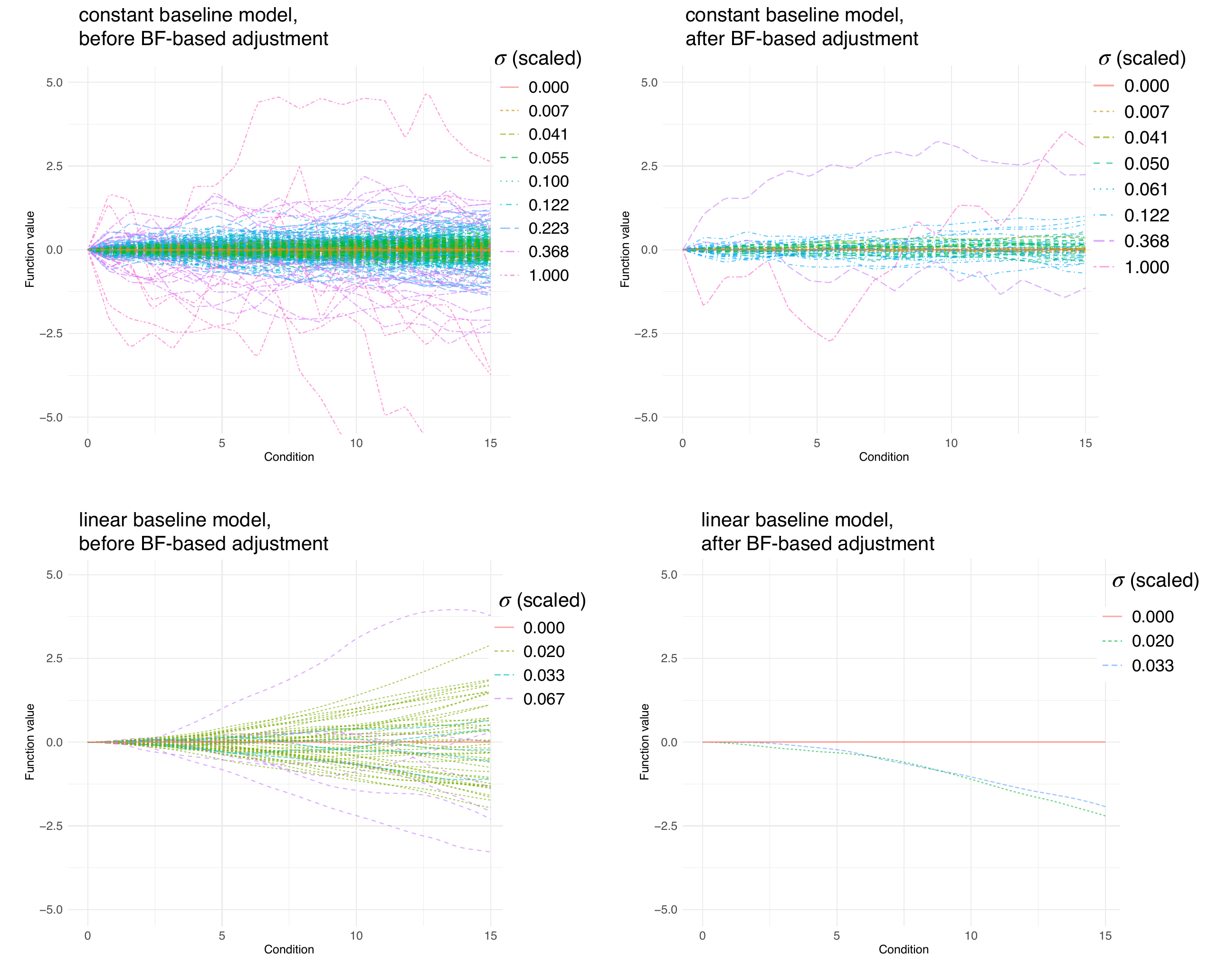}
\caption{
Sample paths randomly drawn from the fitted FASH priors that were
estimated from the dynamic eQTL data in \cref{sec-application}, for $L
= D^1$ (top) and $L = D^2$ (bottom), before (left) and after (right)
the BF-based adjustment.
%
% The baseline models corresponds to constant functions (upper panel)
% and linear functions (lower panel). Sample paths from different
% mixture components are displayed in different colors and line types.
%
For purposes of visualization, we constrained the sample paths to have
an initial condition of zero; that is, $\beta(0) = 0$ in the top row,
and $\beta(0) = \beta'(0) = 0$ in the bottom row. 5,000 randomly
sampled paths are shown in each plot.}
\label{fig:fitted_FASH_priors}
\end{figure}

When using an $L\text{-GP}({\beta};\sigma)$ prior for ${\beta}$, the
choice of $L$ governs the baseline model toward which shrinkage
occurs, and choice of $\sigma$ governs the level of shrinkage. As
$\sigma$ decreases toward zero, the $L$-GP becomes increasingly
constrained to remain close to the baseline
model. Figure \ref{fig:fitted_FASH_priors} illustrates how different
settings of $\sigma$ result in different prior distributions on effect
functions. 

For a flexible family of prior distributions that can appropriately
adapt the amount of shrinkage to the data, we use mixtures of $L$-GP
priors,
\begin{equation}
\label{equ:fash_adaptive} 
\mathcal{G}_{\beta}
= \left\{g_{\beta} : g_{\beta} = \sum_{k=0}^K \pi_k \,
L\text{-GP}({\beta}; \sigma_k)\right\},
\end{equation}
in which the $\{\sigma_k\}_{k=0}^K$ denote a fixed grid of values
ordered from small ($\sigma_0=0$, corresponding to no deviation from
$S_0$) to large ($\sigma_K$), and the $\pi_k$ denote mixture weights
($\pi_k \geq 0, \sum_{k=0}^K \pi_k = 1$). We refer
to \eqref{equ:fash_adaptive} prior family as the ``functional adaptive
shrinkage prior'' family as its construction mirrors the ``adaptive
shrinkage priors'' from \citet{ash,
mash}. 

Since the priors of the form \eqref{equ:fash_adaptive} are fully
specified by the mixture weights $\boldsymbol{\pi} \colonequals
(\pi_k)_{k=0}^K$, learning the prior
$g_{\beta} \in \mathcal{G}_{\beta}$ reduces to learning the prior
weights:
\begin{equation}
\begin{aligned}
\label{eb_pi}
\hat{\boldsymbol{\pi}} =&\;
\mathrm{argmax}_{\boldsymbol{\pi}} \, l(\boldsymbol{\pi}), \\
l(\boldsymbol{\pi}) \colonequals&\;
\sum_{j=1}^J \log \left\{\sum_{k=0}^K \pi_k \ell_{jk}\right\},
\end{aligned}
\end{equation}
where $\ell_{jk}$
% \begin{equation*}
% \ell_{jk} = p_k(\boldsymbol{\hat{\beta}}_j \mid \boldsymbol{s}_j)
% \end{equation*}
denotes the marginal likelihood of unit $j$ under the $k$th component
of the prior; that is, $\ell_{jk}$ is the marginal likelihood
$p(\hat{\bm\beta}_j \mid \boldsymbol{s}_j) = \int
p(\hat{\bm\beta}_j \mid \boldsymbol{\beta}_j, \boldsymbol{s}_j)
\, p_j(\boldsymbol{\beta}_j; g_{\beta}) \, d{\bm\beta}_j$ with a prior of the
form \eqref{equ:fash_adaptive} in which $\pi_k = 1$ and the remaining
mixture weights are zero.
 
% To adaptively learn the amount of shrinkage for each function $\beta_j$, we estimate the prior weights $\boldsymbol{\pi} = (\pi_k)_{k=0}^K$ via empirical Bayes. Let
% \[
% \mathbf{L}_{jk} = p_k(\boldsymbol{\hat{\beta}}_j \mid \boldsymbol{s}_j)
% \]
% denote the marginal likelihood of unit $j$ under the $k$th $L$GP component. Then the EB estimate is
% \begin{equation}\label{eb_pi}
%   \hat{\boldsymbol{\pi}} = \arg \max_{\boldsymbol{\pi}} l(\boldsymbol{\pi)}, \quad
%   l(\boldsymbol{\pi}) = \sum_{j=1}^{J} \log \!\left(\sum_{k=0}^K \pi_k \mathbf{L}_{jk}\right).
% \end{equation}
% The resulting functional adaptive shrinkage prior is
% \begin{equation}\label{eb_prior}
%   \hat{g}_\beta = \sum_{k=0}^K \hat{\pi}_k L\text{GP}(\sigma_k).
% \end{equation}

\subsection{Posterior Inference in FASH}
\label{subsec-inference}

In addition to computing posterior means, variances and other
posterior moments of effect functions, we also use the posterior
distributions \eqref{posterior_betaj} to perform hypothesis testing on
effect functions. All these involve computing expectations with respect to 
posterior distributions on effect functions \eqref{fash_posterior}.
Due to conjugacy of the likelihood \eqref{fash_likelihood} and
the prior \eqref{equ:fash_adaptive}, the posterior distributions are 
also mixtures:
\begin{equation}
\label{posterior_betaj}
p({\beta}_j \mid \hat{\bm\beta}_j, 
\boldsymbol{s}_j, \hat{\boldsymbol{\pi}})
= \sum_{k=0}^K \tilde{\pi}_{jk} \, 
p_k({\beta}_j \mid \hat{\bm\beta}_j, \boldsymbol{s}_j),
\end{equation}
where
$p_k({\beta}_j \mid \hat{{\bm\beta}_j}, \boldsymbol{s}_j)$
denotes the posterior of unit $j$ under the $k$th $L$-GP component,
and the $\tilde{\pi}_{jk}$ denote the posterior mixture weights,
\begin{equation}
\label{posterior_weight}
\tilde{\pi}_{jk} \colonequals
\frac{\hat{\pi}_k \, \ell_{jk}}
     {\sum_{k'=0}^K \, \hat{\pi}_{k'} \ell_{jk'}}.
\end{equation}
Also, the Markovian structure of the $L$-GP facilitates efficient
computation of these posteriors. We elaborate on this in the
Appendix \ref{subsec:computation} of supplement.

\subsubsection{Hypothesis Testing}
\label{sec:hypothesis_testing}

To test $H_0 : {\beta}_j \in S_0$, a common
approach is to control the false discovery rate (FDR) across the $J$
units, either in the classical frequentist
formulation \citep{benjamini1995controlling}, or in closely related
Bayesian formulations \citep{storey2002direct, efron2008microarrays}.
For a subset $\Gamma \subseteq [J]$ flagged as significant, the
estimated FDR is obtained from the local false discovery rate
(lfdr) \citep{efron2008microarrays} as
\begin{equation}\label{equ:fdr}
\widehat{\mathrm{FDR}}(\Gamma) = 
\frac{1}{|\Gamma|} \sum_{j \,\in\, \Gamma} \mathrm{lfdr}(j),
\end{equation}
in which
\begin{equation}
\label{equ:lfdr_def}
\mathrm{lfdr}(j) \colonequals 
p({\beta}_j \in S_0 \mid \boldsymbol{\hat{\beta}}_j, 
\boldsymbol{s}_j, \hat{\boldsymbol{\pi}}).
\end{equation}
Since the $k = 0$ component of the mixture
prior \eqref{equ:fash_adaptive} corresponds to $S_0$, in FASH the lfdr
reduces to
\begin{equation}
\mathrm{lfdr}(j) = \tilde{\pi}_{j0}.
\end{equation}

% By \cref{equ:lfdr_def}, we have
% \begin{equation}\label{equ:lfdr_pi}
%     \text{lfdr}_j = \tilde{\pi}_{j0} = \frac{\hat{\pi}_0 \mathbf{L}_{j0}}{\sum_{k=0}^K \hat{\pi}_k \mathbf{L}_{jk}}.
% \end{equation}

For testing hypotheses of the form $H_0 : \mathcal{F}(\beta_j) = 0$, where $\mathcal{F}: C^p(\Omega)\to\mathbb{R}$ is a functional,
there is the potential use local false sign rates (lfsr) and false
sign rates (FSR), which are generally more robust than the lfdr and
FDR \citep{ash}. For a two-sided alternative $H_1
: \mathcal{F}({\beta}_j) \neq 0$, the lfsr is
\begin{equation}
\label{equ:lfsr}
\mathrm{lfsr}(j) = \min \big\{
p(\mathcal{F}({\beta}_j) \geq 0 \mid \hat{\boldsymbol\beta}_j, 
\boldsymbol{s}_j, \hat{\boldsymbol\pi}),\,
p(\mathcal{F}({\beta}_j) \leq 0 \mid \hat{\boldsymbol\beta}_j, 
\boldsymbol{s}_j, \hat{\boldsymbol\pi}) \big\}.
\end{equation}
For a one-sided alternative, e.g., $H_1:\mathcal{F}({\beta}_j) >
0$, we define the lfsr as
\begin{equation}
\label{equ:lfsr_one}
\mathrm{lfsr}(j) = 
p(\mathcal{F}({\beta}_j) \leq 0 \mid 
\hat{\bm\beta}_j, \boldsymbol{s}_j, \hat{\boldsymbol\pi}).
\end{equation}
The cumulative FSR can be computed from the lfsr analogously
to the FDR \citep{ash}.

\subsection{BF-based Adjustment of $\hat{\pi}_0$}
\label{subsec-BFcontrol}

The lfdr and FDR can be quite sensitive to the estimate of the null
proportion, $\pi_0$. Although lfsr and FSR are generally more robust than lfdr and FDR, they can still be influenced by the estimate $\hat{\pi}_0$. In particular, if $\hat{\pi}_0$ falls below the true
$\pi_0$, inference can become anti-conservative, resulting in an
inflated FDR or FSR. Consequently, a conservative estimate, in which
$\hat{\pi}_0 \ge \pi_0$, would be preferred.

One major reason for underestimating $\hat\pi_0$ in FASH is
misspecification of the prior family; if the mixture prior
family \eqref{equ:fash_adaptive} is not sufficiently flexible to
approximate the true marginal distribution, then the
maximum-likelihood estimate of $\pi_0$ can become inaccurate. In the
simpler univariate setting \citep{ash}, prior misspecification may be
a mild concern, but in FASH, like in other multivariate
settings \citep{mash, liu2024flexible}, the concern is greater due to
the difficulty of adequately modeling data in high dimensions.
%
% Appendix~C provides a detailed discussion of how misspecification of
% the alternative distribution under $H_1$ affects $\hat{\pi}_0$.
%

To guard against this issue, we describe a simple yet effective
BF-based adjustment of $\hat{\pi}_0$. {\color{changed_color} This
adjustment is not specific to FASH, and could potentially be used in
other settings where prior could be misspecified.} This adjustment only
exploits the fact that, under the null hypothesis, the BF in favor of
the alternative has an expectation of 1 (that is, BFs are
$e$-variables; see \citealt{vovk2021values}). This property of BFs is
stated more formally in the following lemma.
\begin{lemma}[BF-Moment]\label{lemma:BF_moment} 
\rm Let the Bayes factor for unit $j$ be
\begin{equation*}
\text{BF}_j \colonequals \frac{p_{1}(\hat{\boldsymbol\beta}_j \mid 
\boldsymbol{s}_j)}{p_{0}(\hat{\boldsymbol\beta}_j \mid \boldsymbol{s}_j)},
\end{equation*}
where $p_1$ and $p_0$ denote the marginal likelihoods under the
alternative ($H_1$) and null ($H_0$) hypotheses, respectively. Under
the null hypothesis, $H_0$, we have $\mathbb{E}_0(\text{BF}_j) = 1$,
regardless of how the alternative hypothesis $H_1$ is specified.
\end{lemma}

\begin{proof}
\begin{equation*}
\mathbb{E}_0(\text{BF}_j) = 
\int \frac{p_1(\hat{\boldsymbol\beta}_j \mid 
\boldsymbol{s}_j)}{p_0(\hat{\boldsymbol\beta}_j \mid 
\boldsymbol{s}_j)} \times p_0(\hat{\boldsymbol\beta}_j \mid 
\boldsymbol{s}_j) \, d\hat{\boldsymbol\beta}_j
= \textstyle 
\int p_1(\hat{\boldsymbol\beta}_j \mid \boldsymbol{s}_j) \, 
d\hat{\boldsymbol\beta}_j = 1.
\end{equation*}
\end{proof}

\begin{algorithm}[t]
\caption{BF-based adjustment of $\pi_0$.}
\label{algo:BF}
\begin{algorithmic}[1]

\REQUIRE A $J \times (K+1)$ matrix $\mathbf{L}$ in which element $(j,k)$ 
is $\ell_{jk}$, the marginal likelihood of unit $j$ under the $k$th 
component of the mixture prior; the estimated prior weights,
$\hat{\boldsymbol{\pi}} = (\hat{\pi}_0, \ldots, \hat{\pi}_K) \in 
\mathbb{R}^{K+1}$; a set of candidate cutoff values, 
$\mathcal{C} \subset \mathbb{R}^{+}$ a ``buffer'' tuning parameter, 
$\epsilon > 0$, typically close to zero.

\STATE Normalize the alternative weights,
$\hat{\pi}_k^{*} = \hat{\pi}_k / \sum_{k'=1}^K \hat{\pi}_k'$, 
$k = 1, \ldots, K$.

\STATE Compute the ``collapsed'' likelihoods,
$\ell_{j0}^c = \ell_{j0}$ and $\ell_{j1}^c
= \sum_{k=1}^K \ell_{jk} \hat{\pi}_k^{*}$, $j = 1, \ldots, J$.

\STATE Compute the Bayes factors,
$\text{BF}_j = \ell_{j1}^c / \ell_{j0}^c$ for $j = 1, \ldots,
J$.

\FOR{$c \in \mathcal{C}$}

\STATE $J_0 = \sum_{j=1}^J \mathbb{I}\{\text{BF}_j < c\}$

\STATE $\hat{\pi}_0(c) = J_0/J$.

\STATE $\mu(c) = \sum_{j=1}^J 
\text{BF}_j \times \mathbb{I}\{\text{BF}_j < c\}/J_0$.

\ENDFOR

\STATE $c^{*} = \inf\{c \in \mathcal{C} : \mu(c) \geq 1+\epsilon\}$.

\STATE Adjust the null prior weights, $\hat{\pi}_0 = \hat{\pi}_0(c^*)$.

\STATE Adjust the alternative prior weights,
$\hat{\pi}_k = \hat{\pi}_k^{*} (1 - \hat{\pi}_0(c^{*}))$,
$k=1,\ldots,K$.

\RETURN the adjusted weights, 
$\hat{\boldsymbol{\pi}} = (\hat{\pi}_0, \ldots, \hat{\pi}_K)$.

\end{algorithmic}
\end{algorithm}

Let $J_0$ and $J_1$ denote the numbers of null and alternative units,
respectively, with $J = J_0 + J_1$ and $\pi_0 = J_0/J$.  By the law of
large numbers, when $J_0$ is large we have
\begin{equation*}
\frac{1}{J_0} \sum_{j\, \in\, \mathcal{H}_0} \text{BF}_j \;\approx\; 1,
\end{equation*}
where $\mathcal{H}_0 \subseteq \{1, \ldots, J\}$ denotes the null
{units}. This observation motivates a simple conservative procedure to
estimate $J_0$, and hence $\pi_0$: seek the largest possible set of
units that is, on average, consistent with being all null (average BF
$\leq 1$).  Algorithm~\ref{algo:BF} implements this procedure, which
we call the {\em BF-based adjustment of $\pi_0$}. The
following \cref{thrm:BF_control} establishes the
conservativeness of the adjusted $\hat{\pi}_0$.
\begin{theorem}[BF-based adjustment gives conservative estimate]
\label{thrm:BF_control}
\rm Assume $\hat{\pi}_0$ is obtained from Algorithm~\ref{algo:BF}, and
the $J_0$ null effects are iid from the null distribution specified by
the prior. Then for any alternative distribution, and for any
$\epsilon > 0$,
\begin{equation*}
\hat{\pi}_0 \geq \pi_0 \; \text{almost surely as} \; J_0 \to \infty.
\end{equation*}
\end{theorem}
The proof of this theorem is given in Appendix \ref{sec:proofs} of the
supplement.

Importantly, {\em \cref{thrm:BF_control} requires only that the null
distribution be correctly specified} (the $k = 0$ mixture component of
the FASH prior, $g_{\beta}$), whereas the alternative distribution (mixture
components 1 through $K$ in $g_{\beta}$) may be misspecified. For
example, this result does not depend on the shrinkage values
$\sigma_1, \ldots, \sigma_K$, which means that the BF-based adjustment
should work even when a coarse grid of values is used to reduce
computation.

The requirement that the null hypothesis be correctly specified is
analogous to the usual conditions required for calibration of
classical $p$-values, and is much less onerous than the requirement
that the full distribution be correctly specified.  Nonetheless, in
this setting the requirement is not entirely benign, since the null
distribution is not simply a point mass at zero, but rather a diffuse
distribution over the function class $S_0$ of dimension $p \ge
1$. (While diffuse priors can be a problem for Bayes factor
computation, as they can make the marginal likelihood arbitrarily
scaled, here the diffuse component is shared across all mixture
components, so these arbitrary scales cancel out when taking their
ratio, and the resulting Bayes factors remain valid;
see \citealt{BFANOVA, servin2007imputation}.)  If the true null
functions are not diffuse over $S_0$, then technically this introduces
prior misspecification and the conditions of \cref{thrm:BF_control}
fail to hold.  Nonetheless, in our experiments the BF-based
adjustments produced consistently conservative estimates of
$\hat{\pi}_0$, even when the prior is misspecified
(see Appendix \ref{sec:simulations} of the supplement).

{\color{changed_color} From \eqref{posterior_weight}
and \eqref{equ:lfdr_def}, overestimating $\hat{\pi}_0$ in the FASH
prior shifts posterior mass toward the baseline model, $S_0$, yielding
more conservative decisions against $H_0$. 
See \cref{fig:fitted_FASH_priors} for an illustration of
this. In Appendix \ref{sec:proofs} of the supplement, we provide
further details by examining the posterior odds comparing the null and
the alternative. The conservative behavior of related quantities, such
as the lfdr and lfsr, is also investigated empirically in the
Appendix \ref{sec:simulations} of the supplement.}

\section{Analysis of Dynamic eQTLs in the iPSC Cardiomyocyte Differentiation Study}
\label{sec-application}

{\color{changed_color} We evaluated FASH on a variety of performance
measures in simulated data sets. These simulation experiments are
described in detail in Appendix \ref{sec:simulations} of the
supplement. Of note, our simulations confirmed that the BF-based
adjustment produces conservative estimates of $\pi_0$, and therefore
produces conservative estimates of lfdr and FDR, even when
the prior is misspecified. Reassuringly, these
conservative estimates did not subtantially reduce power.}

Now we focus on a real data application, a reanalysis of dynamic
eQTLs in the iPSC cardiomyocyte differentiation study
\citep{strober2019dynamic}. This study measured daily gene
expression over a 16-day period in induced pluripotent stem cells
(iPSCs) derived from 19 Yoruba HapMap cell lines undergoing
differentiation into cardiomyocytes. The goals of the analysis were to
identify the dynamic eQTLs---that is, to identify genetic loci whose
effects on gene expression vary over time---and to characterize how
genetic regulation of gene expression changes throughout the process
of differentiation. We show that FASH is able to meet these goals in a
principled way, and we compare the FASH results to the original
analysis based on parametric interaction models.
%
% involved {\em ad hoc} solutions that limited discovery and
% complicated interpretation of the results.
%
% \subsection{Data from the iPSC Cardiomyocyte Differentiation Study}

After carrying out the data preprocessing and quality control
procedures from \cite{strober2019dynamic}, we obtained data on $J
= \mbox{1,009,173}$ gene-variant pairs (6,362 genes). All the genetic
variants considered in this analysis were SNPs within 50 kb of the
gene's transcription start site (TSS).  For each gene-variant pair
$j \in [J]$, we obtained eQTL effect estimates $\hat{\beta}_j(t_r)$ at
16 time points by fitting linear regression models separately for each
of the time points $r \in [16]$. The standard errors $s_{jr}$ of the
eQTL effect estimates were subsequently adjusted to address concerns
about inflation of type I errors due to the small sample size (see
Appendix \ref{subsec:application-details} of the supplement for
details).
% 
% After this adjustment, we defined the estimated
% effects and their standard errors for each gene-variant pair as
% \begin{equation}
% \begin{aligned}
% \boldsymbol{\hat{\beta}_j} &= 
% (\hat{\beta}_j(t_1), \ldots, \hat{\beta}_j(t_{R})]^T, \\
% \boldsymbol{\tilde{s}}_j &= 
% (\tilde{s}_{j1}, \ldots, \tilde{s}_{jR}]^T, \quad j \in [J],\\
% \hat{\beta}_j(t_r)\ |\ \beta_j(t_r) 
% &\overset{\text{iid}}{\sim} 
% N(\beta_j(t_r), \tilde{s}_{jr}^2), \quad r \in [R].
% \end{aligned}
% \end{equation}

Specifically, we considered two inference aims:
\begin{enumerate}

\item {\em Identify the dynamic eQTLs.} This was implemented in FASH 
using mixtures of $L$-GP priors with $L = D^1$. We refer to the
fitted model in this analysis as the ``FASH-$\text{IWP}_1$'' model.
% 
%  which tests the null hypothesis $H_{0j}: \beta_j = c_j$ for some
% constant $c_j$ We set \( L = \frac{d}{dt}$ ($\text{IWP}_1$).
% 

\item {\em Identify the nonlinear dynamic eQTLs.} This was implemented 
in FASH using mixtures of $L$-GP priors with $L = D^2$. We refer to
the fitted model in this analysis as the ``FASH-$\text{IWP}_2$''
model.
% 
% $j$th variant-gene pair, test the null hypothesis $H_{0j}: \beta_j =
% c_j + b_j t$ for some constants $c_j$ and $b_j$.  We set $L
% = \frac{d^2}{dt^2}$ ($\text{IWP}_2$).
%

\end{enumerate}

% For each of the two inference aims, we employ the corresponding FASH
% model described in \cref{sec-method}.
%
% \begin{equation}\label{equ:fash_eqtl}
%     \begin{aligned}
%         \hat{\beta}_j(t_r)\,|\,\beta_j(t_r) &\overset{\text{ind}}{\sim} N\left(\beta_j(t_r), \tilde{s}_{j,r}^2\right), \\
%         \beta_j\,|\,\boldsymbol{\pi} &\overset{\text{iid}}{\sim} \sum_{k=0}^K \pi_k\, \text{GP}(\sigma_k), \quad \forall j \in [J],\, r \in [R].
%     \end{aligned}
% \end{equation}
% Each component \( \text{GP}(\sigma_k) \) in the mixture is an \( L \)-GP prior defined by an integrated Wiener process (IWP), where the linear differential operator \( L \) encodes the baseline model. 

All the results were generated using version 0.1.42 of our R package,
\texttt{fashr}. R code implementing our analyses is available at
\url{https://github.com/stephenslab/fashr-paper}. 
In all of the analyses, the FASH priors were defined on an equally
spaced grid on the log-scale; {see
Appendix \ref{subsec:application-details} for details.  After
estimating the mixture weights $\boldsymbol{\pi}$ by
maximum-likelihood, we applied the BF-based adjustment, then we used
the adjusted weights $\hat{\bm\pi}$ to make inferences from the
posterior distributions $p(\beta_j \mid \hat{\bm\beta}_j, {\bm
s}_j, \hat{\bm\pi})$. We used an FDR threshold of $\alpha
= \mbox{0.05}$ for identifying ``significant'' dynamic
eQTLs.\footnote{We use FDR so that the FASH analysis is more
comparable to the original analysis, but lfdr is generally preferred;
see \cref{discussion} for further discussion on the use of FDR and FSR
vs. lfsr and lfsr for hypothesis testing in FASH.} Performing each of
the FASH analyses end-to-end using \texttt{fashr} took about 11
hours. The computations were performed on Linux machines (Scientific
Linux 7.4) with Intel Xeon Gold 6248R (``Cascade Lake'') processors
and 16 parallel threads.

\begin{figure}[!t]
\centering
\includegraphics[width=0.925\textwidth]{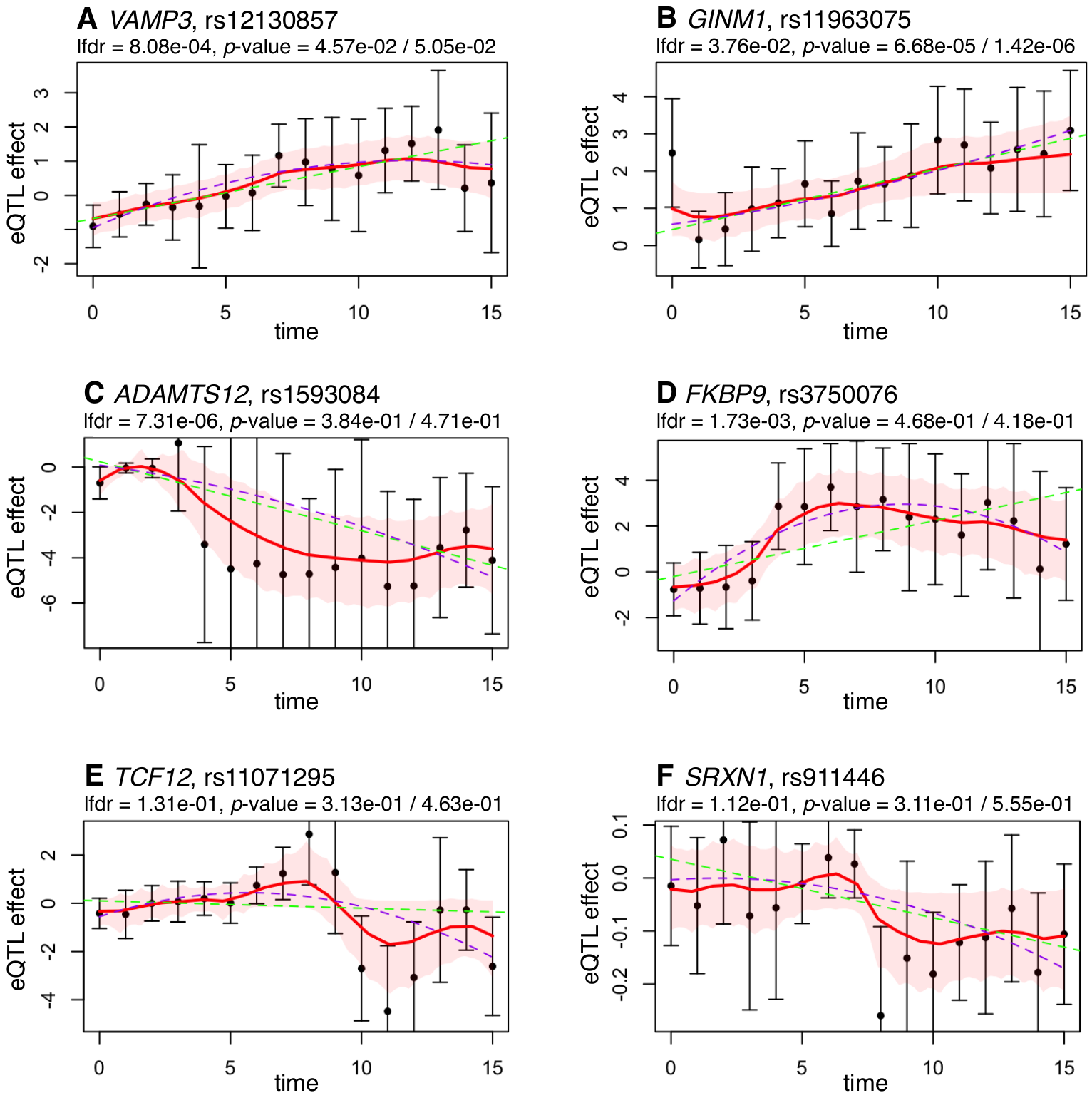}
\caption{
Examples of dynamic eQTLs identified by FASH ($L = D^1$), and
comparison with parametric modeling approach
\citep{strober2019dynamic}.  In each plot, the posterior mean
of $\beta_j$ is shown as a red solid line, and the 95\% credible
interval is depicted by the shaded region. Observed effect size
estimates $\hat{\beta}_j$ are shown as black dots, with vertical error
bars representing $\pm 2$ the (adjusted) standard errors
$\tilde{s}_{jr}$. The lfdr is the lfdr from the FASH-$\text{IWP}_1$
model after the BF adjustment; the {\em p-}values are from the
analyses of \cite{strober2019dynamic} with linear/quadratic
interaction models. The inverse-variance weighted least squares
estimates for the linear ($G_c \times t$) and quadratic ($G_c \times
t^2$) parametric interaction models are shown as dashed lines (green =
linear, purple = quadratic).}
\label{fig:dynamic_compare}
\end{figure}

\subsection{Discovery of Dynamic eQTLs}

Examining in detail some of the dynamic eQTLs identified by FASH with
$L = D^1$ illustrates the variety of dynamics that can be captured by
FASH.
%
% and provides some insight into the differences between FASH and
% the original analysis that used parametric interaction models.
%
The examples in the top row of Fig. \ref{fig:dynamic_compare} (A, B)
have effects that appear to gradually become stronger over time.  In
these examples, both the linear interaction ($G_c \times t$, where
$G_c$ is the SNP genotype) and quadratic interaction ($G_c \times
t^2$) models used in \cite{strober2019dynamic} also provided good fits
to the eQTL effects in these examples, and indeed these gene-variant
pairs were identified as dynamic eQTLs
in \cite{strober2019dynamic}. By contrast, the examples in the middle
and bottom rows of Fig. \ref{fig:dynamic_compare} (C--F) were
identified as dynamic eQTLs by FASH, but they were not identified in
the original analysis. Indeed, in all these examples, the linear and
quadratic models appear to be a poor fit for the nonlinear dynamics of
these eQTLs; for example, in C and D the effect strengthens rather
abruptly around day 5, and in E and F the suddenly gets stronger
at around day 10. Dynamic eQTLs C and D were very strongly identified
by FASH (lfdr of 0.01 or lower)
%
% and show very clear changes over time.
%
E and F are ``borderline'' dynamic eQTLs (lfdr > 0.1) with more subtle
temporal dynamics. Additional illustrative examples, including
examples of gene-variant pairs that FASH did not identify as dynamic
eQTLs, are given in Figures \ref{fig:dynamic_highlight}, \ref{fig:nonsig_order1} and \ref{fig:missFASH} in the supplement.

\begin{figure}[!t]
\centering
\includegraphics[width=0.825\textwidth]{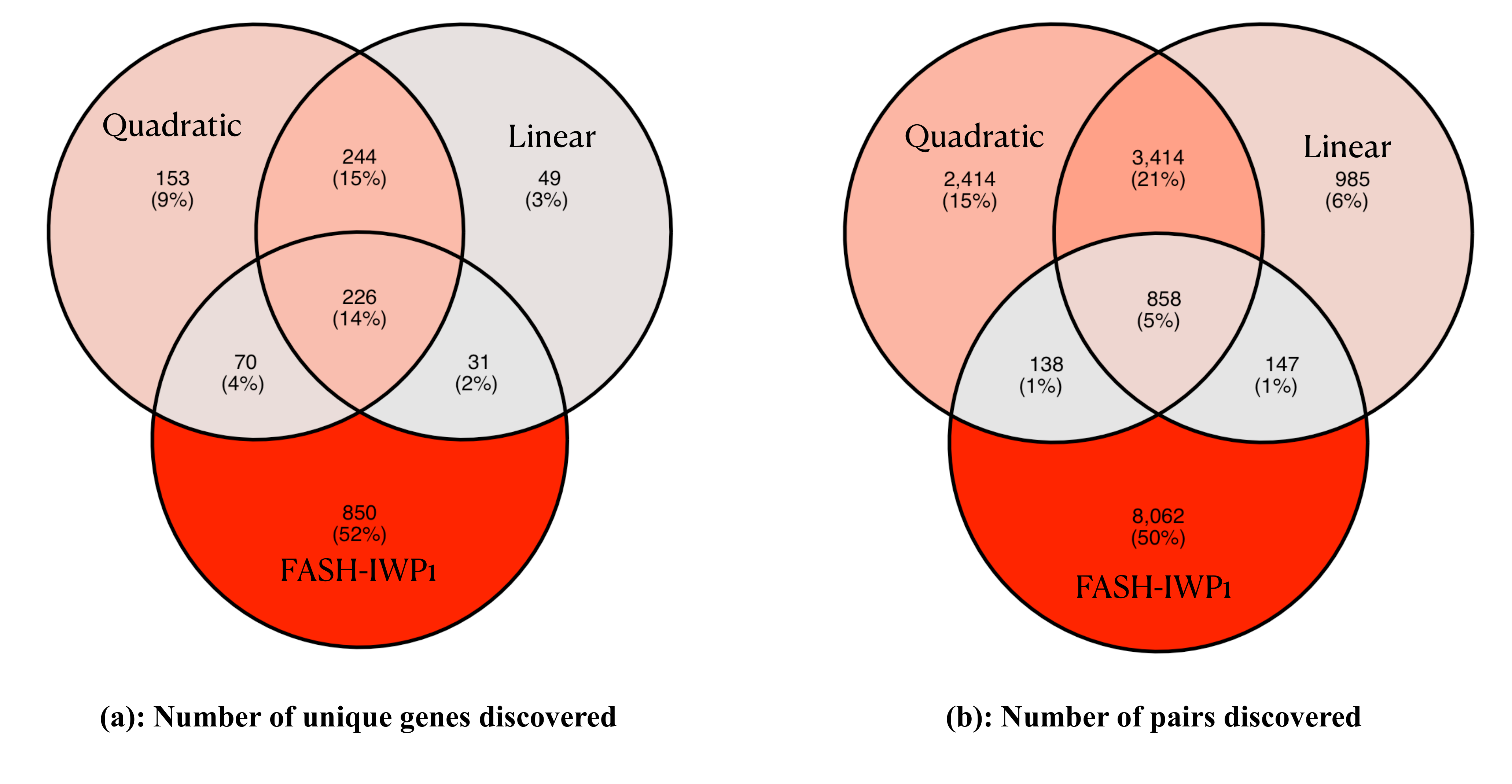}
\caption{Discovery of
dynamic eQTLs using FASH vs. linear/quadratic parametric interaction
models: (a) genes and (b) gene-variant pairs. }
\label{fig:dynamic_vienn}
\end{figure}

As a result, given FASH's ability to capture diverse temporal
patterns, FASH identified many more dynamic eQTLs than the original
analysis (\cref{fig:dynamic_vienn}): at an FDR threshold of $\alpha =
0.05$, FASH-{$\text{IWP}_1$} identified dynamic eQTLs at
%
% 43,860 gene-variant pairs in 3,258
% genes before the BF adjustment, and
%
9,205 gene-variant pairs (in 1,177 genes), or about 1\% of all
gene-variant pairs tested (19\% of all genes); at an eFDR threshold of
0.05, \cite{strober2019dynamic} identified dynamic eQTLs at 5,404
gene-variant pairs (in 550 genes) using the linear interaction model,
and 6824 gene-variant pairs (693 genes) using the quadratic
interaction model. Despite the increased discovery of dynamic eQTLs,
the FASH inferences are still ``conservative'' in that they were obtained
using the BF-based adjustment of the FASH prior. (The effect of this
adjustment on the FASH priors was shown
in \cref{fig:fitted_FASH_priors}.)

It should be noted that, since there were several other differences in
the two analyses beyond the increased flexibility of FASH, it is
likely that these differences also contributed to differences in
discovery, and may partly explain why many of the dynamic eQTLs
identified by the linear and quadratic interaction models were not
reproduced in our analysis. For example, recognizing that the small
sample sizes may lead to inflated type I errors, we adjusted the
standard errors following a simple procedure described in the
supplement, whereas the original analysis did not account for this
issue. So it is possible that not accounting for this issue in the
original analysis lead to a greater number of false positives. Another
important consideration is that the two analyses took very different
approaches to estimating the false discovery rate: FASH, by taking an
empirical Bayes approach, estimates the (prior) null proportion as
part of the model fitting, which is then used for the lfdr and FDR
estimation; the analysis of \cite{strober2019dynamic} used a
permutation-based approach to estimate the null distribution of {\em
p-}values, then applied the approach of \cite{gamazon2013integrative}
to estimate the FDR. 
See \Cref{fig:lfdr_pval,fig:FDR_eFDR} in the supplement for more detailed comparisons of the two analyses, comparing the FASH lfdr
values (and their corresponding FDR estimates) with the {\em p-}values (and the corresponding eFDR estimates) from the original analysis.

\subsection{Characterization of Dynamic eQTLs}

Now we move from discovery of dynamic eQTLs to characterizing the
dynamic changes of the dynamic eQTLs; in particular, we would like to 
characterize {\em how} the genetic effects on gene expression evolve
throughout the process of cell differentiation.
\cite{strober2019dynamic} also sought to characterize the
dynamic eQTLs, but their statistical analysis was complicated by the
limitations of the available methods. Here we show that this is much
more straightforwardly accomplished within the FASH modeling
framework, as well as being more valid because it is accompanied by
(appropriately calibrated) posterior statistics such as lfdrs and
lfsrs.

\begin{figure}[!t]
\centering
\includegraphics[width=0.925\textwidth]{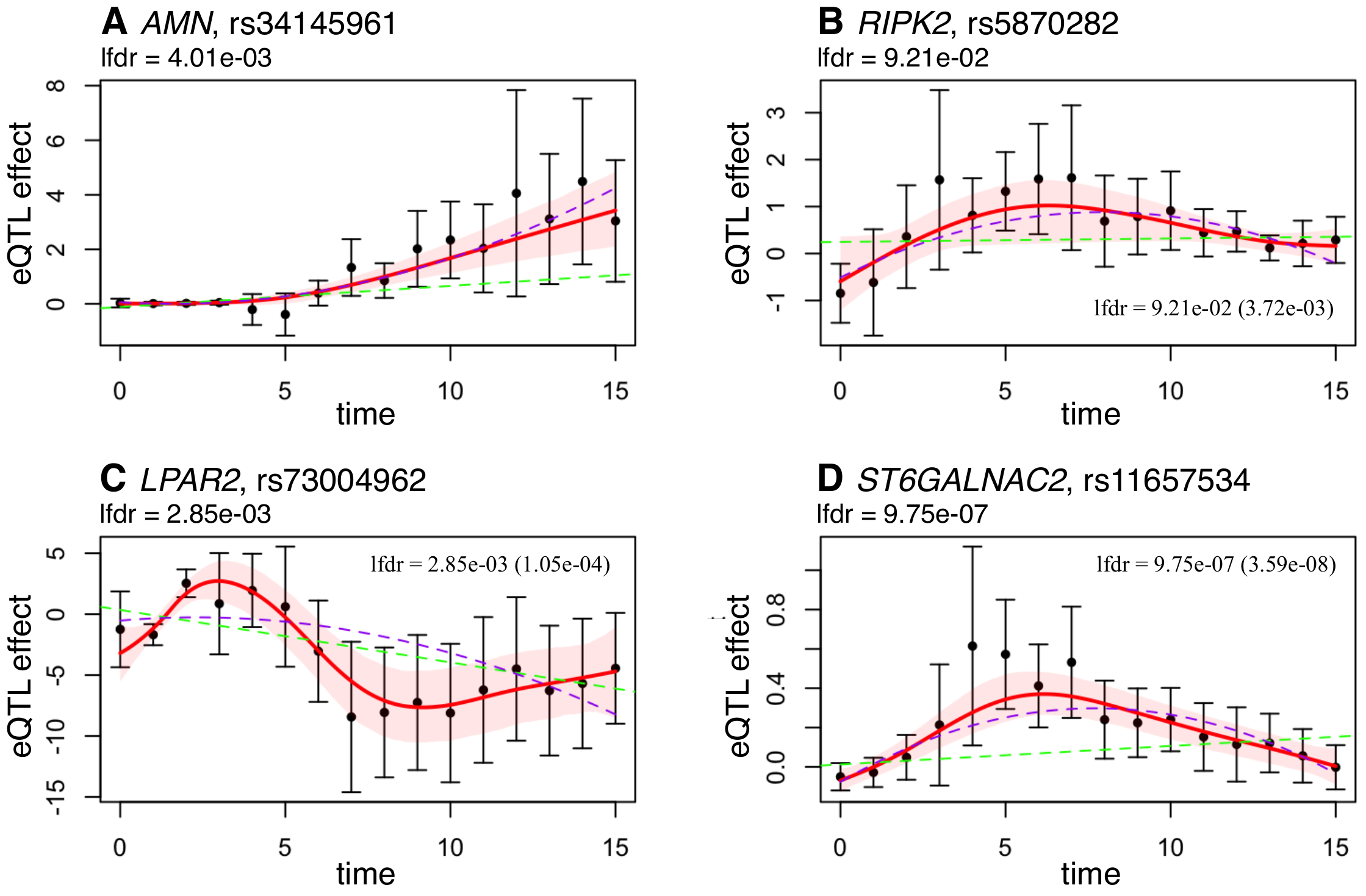}
\caption{Examples of nonlinear dynamic eQTLs identified by FASH, with $L = D^2$.
%
% at a false discovery rate threshold of $5\%$. 
%
% Plot elements are labeled as in \cref{fig:dynamic_highlight}.
% As in \cref{fig:dynamic_highlight}, the posterior mean
%
In each plot, the posterior mean of $\beta_j$ is shown as a red solid
line, and the 95\% credible interval is depicted by the shaded
region. Observed effect size estimates $\hat{\beta}_j$ are shown as
black dots, with vertical error bars representing $\pm 2$ the
(adjusted) standard errors $\tilde{s}_{jr}$. The lfdr is the lfdr from
the FASH-$\text{IWP}_2$ model after the BF adjustment. The
inverse-variance weighted least squares estimates for the linear
($G_c \times t$) and quadratic ($G_c \times t^2$) parametric
interaction models are shown as dashed lines (green = linear, purple =
quadratic).}
\label{fig:nonlinear_highlight}
\end{figure}

By reanalyzing the data with a linear baseline model instead of a
constant baseline model---that is, using a FASH prior with $L =
D^2$---FASH identifies the gene-variant pairs with {\em nonlinear}
dynamic effects on gene expression. At an FDR threshold of $\alpha =
0.05$, this ``FASH-$\text{IWP}_2$'' model found a small number
nonlinear dynamic eQTLs: 44 gene-variant pairs within 9 genes
(Table \cref{tab:eqtl_categories}). (Without the BF-based adjustment,
FASH-$\text{IWP}_2$ identified 159 gene-variant pairs within 37 genes
at $\alpha = 0.05$, still a much smaller number than the overall
number of genes with dynamic eQTLs; see \cref{fig:dynamic_vienn}.)

A few examples of these nonlinear dynamic eQTLs are shown
in \cref{fig:nonlinear_highlight}.  (See also \cref{fig:nonsig_order2}
for examples of dynamic eQTLs that were not classified as
``nonlinear''.)  The diversity of the nonlinear effects is quite
striking. Clearly, the quadratic interaction model, even though it was
intended for identifying nonlinear dynamic effects
in \cite{strober2019dynamic}, is not flexible enough to capture the
full range of nonlinear effects.

The ability of FASH to test hypotheses of the form $\mathcal{F}(\beta)=0$ or $\mathcal{F}(\beta)>0$ (\cref{sec:hypothesis_testing}) for an arbitrary functional $\mathcal{F}$ provides new ways to characterize dynamic effects in a systematic fashion.
For example, \cite{strober2019dynamic} were interested in identifying the
dynamic eQTLs with effects on expression that switch direction. 
To frame this as a hypothesis test, we define a functional $\mathcal{F}(\beta)$ that measures whether the genotype effect exhibits a sign-changing ``switch'' with magnitude exceeding a threshold $c>0$:
\begin{equation}
\mathcal{F}(\beta) =
\min \left\{ \max_t \beta^{+}(t),\,
\max_t \beta^{-}(t) \right\} - c,
\end{equation}
where $\beta^{+}$ and $\beta^{-}$ denote the positive and negative parts of $\beta$, respectively, so that $\beta = \beta^{+} - \beta^{-}$.
When $\mathcal{F}(\beta)>0$, there exist time points $t_{+}$ and $t_{-}$ such that $\beta(t_{+})>c$ and $\beta(t_{-})<-c$, implying an effect difference of at least $2c$ between these time points for a one-allele change in genotype.
For $c=0.25$, this corresponds to a minimum difference of $4c=1$ between the two-homozygote genotypes.

\begin{table}[!t]
\centering
\begin{tabular}{cccl}
\toprule 
category & gene-variants & genes & functional, $\mathcal{F}(\beta)$ \\
\midrule 
dynamic & 44 & 9 & (not applicable) \\
early  & 124 & 8   &
$\max_{t \,\leq\, 3} |\beta(t)| - \max_{t \,>\, 3} |\beta(t)|$ \\
middle & 24  & 5   &
$\max_{4 \,\leq\, t \,\leq\, 11} |\beta(t)| -
\max_{t \,<\, 4 \text{ or } t \,>\, 11} |\beta(t)|$ \\
late   & 20  & 12  &
$\max_{t \,\geq\, 12} |\beta(t)| -
\max_{t \,<\, 12} |\beta(t)|$ \\
switch & 984 & 250
& $\min\{\max_{0 \,\leq\, t \,\leq\, 15} \beta^{+}(t),
\max_{0 \,\leq\, t \,\leq\, 15} \beta^{-}(t)\} - c$, $c = 0.25$ \\
\bottomrule
\end{tabular}
\caption{
Classification of dynamic eQTLs based on temporal effect patterns:
number of gene-variant pairs and unique genes identified at an FDR
(top row) or FSR (other rows) of 0.05.
%
% Each variant-gene pair is assigned to a category by testing whether
% the corresponding functional $\mathcal{F}(\beta_j) > 0$, as defined
% for each pattern.
%
For the {\em switch} category, setting $c = 0.25$ ensures that the
largest effect size difference across the range of genotype dosages
(0--2) is at least 1, since $2 \times c \times 2 = 1$.
See \cref{fig:examples_classification} for a examples illustrating
each category.}
% 
% For each category, some highlighted enriched HALLMARK gene sets are
% shown in the last column, with the corresponding p-values from
% Fisher’s exact test given in parentheses.
%
\label{tab:eqtl_categories}
\end{table}

\begin{table}[!t]
\centering
\begin{tabular}{lccc}
\toprule 
& genes & {\em p-}value & {\em q-}value \\
Hallmark gene set & all/switch & all/switch & all/switch \\
\midrule 
genes up-regulated in response to hypoxia & 25/11 & 0.017/0.00067 & 0.436/0.025 \\
genes up-regulated by KRAS activation & 11/7 & 0.24/0.0022 & 0.840/0.035 \\
\bottomrule
\end{tabular}
\caption{Gene set enrichment analysis of genes with dynamic eQTLs vs.
genes with switch dynamic eQTLs. The ``genes'' column gives the number
of dynamic genes eQTL or switch dynamic eQTL genes beloning to the
Hallmark gene set \citep{liberzon2015molecular, msigdb, gsea}. GSEA
{\em p-}values and {\em q-}values were computed using the ``enricher''
function in the clusterProfiler R package \citep{clusterprofiler,
clusterprofilerR}.  See \cref{tab:hallmark_enrichment_all_vs_switch}
of the supplement for more detailed GSEA results.}
\label{tab:gsea_mini}
\end{table}

In addition to the {\em switch} category, we used the hypothesis
testing framework to group the dynamic eQTLs into three other
categories:
\begin{itemize}

\item {\em Early:} The strongest effect occurs sometime during the first 
  3 days of cell differentiation.

\item {\em Middle:} The strongest effect occurs sometime between days 4 and 11.

\item {\em Late:} The strongest effect occurs sometime during the final 4 days.

\end{itemize}
Each of these categories corresponds to an inequality of the form
$\mathcal{F}(\beta) > 0$ (\cref{tab:eqtl_categories}). (Note that the
switch category is not exclusive of the other categories; for example,
a dynamic eQTL should be identified as both {\em early} and {\em
switch}.) The numbers of gene-variants and genes assigned to each of
these categories at an FSR threshold of 0.05 (after applying the
BF-based adjustment to the prior) are given
in \cref{tab:eqtl_categories}. Examples of dynamic eQTLs in
each category are shown in \cref{fig:examples_classification}.
%
% we identified 124 variant-gene pairs as early dynamic eQTLs
% (spanning 8 unique genes), 24 as middle dynamic eQTLs (5 genes), 20
% as late dynamic eQTLs (12 genes), and 984 as switch dynamic eQTLs
% (250 genes).
%
Only a very small number of dynamic eQTLs were classified {\em early},
{\em middle} or {\em late}, but many dynamic eQTLs were identified as
having effects that switch direction.  To better understand the
biological significance of the switch dynamic eQTLs, we performed a
gene set enrichment analysis (GSEA) on the Hallmark gene
sets \citep{liberzon2015molecular, msigdb, gsea} using
clusterProfiler \citep{clusterprofiler, clusterprofilerR}, and
compared the GSEA results for the 250 genes with switch dynamic eQTLs
vs. the 1,177 genes with any type of dynamic eQTL. Interestingly, the
top two Hallmark gene sets for the switch dynamic eQTLs (ranked by
{\em p-}value) were genes upregulated in low oxygen levels (i.e.,
hypoxia) and genes upregulated by K-Ras, and these were also among the
top gene sets for all dynamic eQTLs, but the enrichments were much
stronger when considering the switch dynamic eQTLs only
(\cref{tab:gsea_mini}). Both hypoxia and K-Ras are well known to have
strong effects on proliferation and differentiation of stem cells, and
so it potentially significant that the switch dynamic eQTLs are more
strongly enriched for these pathways. More detailed GSEA results are
provided in
\cref{tab:hallmark_enrichment_all_vs_switch} in the supplement.

%
% where $\mathcal{F}$ denotes a task-specific functional
% (see \cref{tab:eqtl_categories}), and therefore can be expressed as
% a posterior inference in the form of \eqref{equ:lfsr_one} to compute
% a lfsr for each gene-variant pair.
% 
% Given the posterior distribution of
% $\beta_j$, we compute the lfsr for each variant-gene pair as
% \begin{equation*}
% \mathrm{lfsr}_j = P[\mathcal{F}(\beta_j) \leq 0 \mid
% \boldsymbol{\hat{\beta}}, \boldsymbol{\hat{\pi}}],
% \end{equation*}
%
% Then we use it to compute and control the cumulative FSR.

\begin{figure}[!t]
\centering
\includegraphics[width=0.925\textwidth]{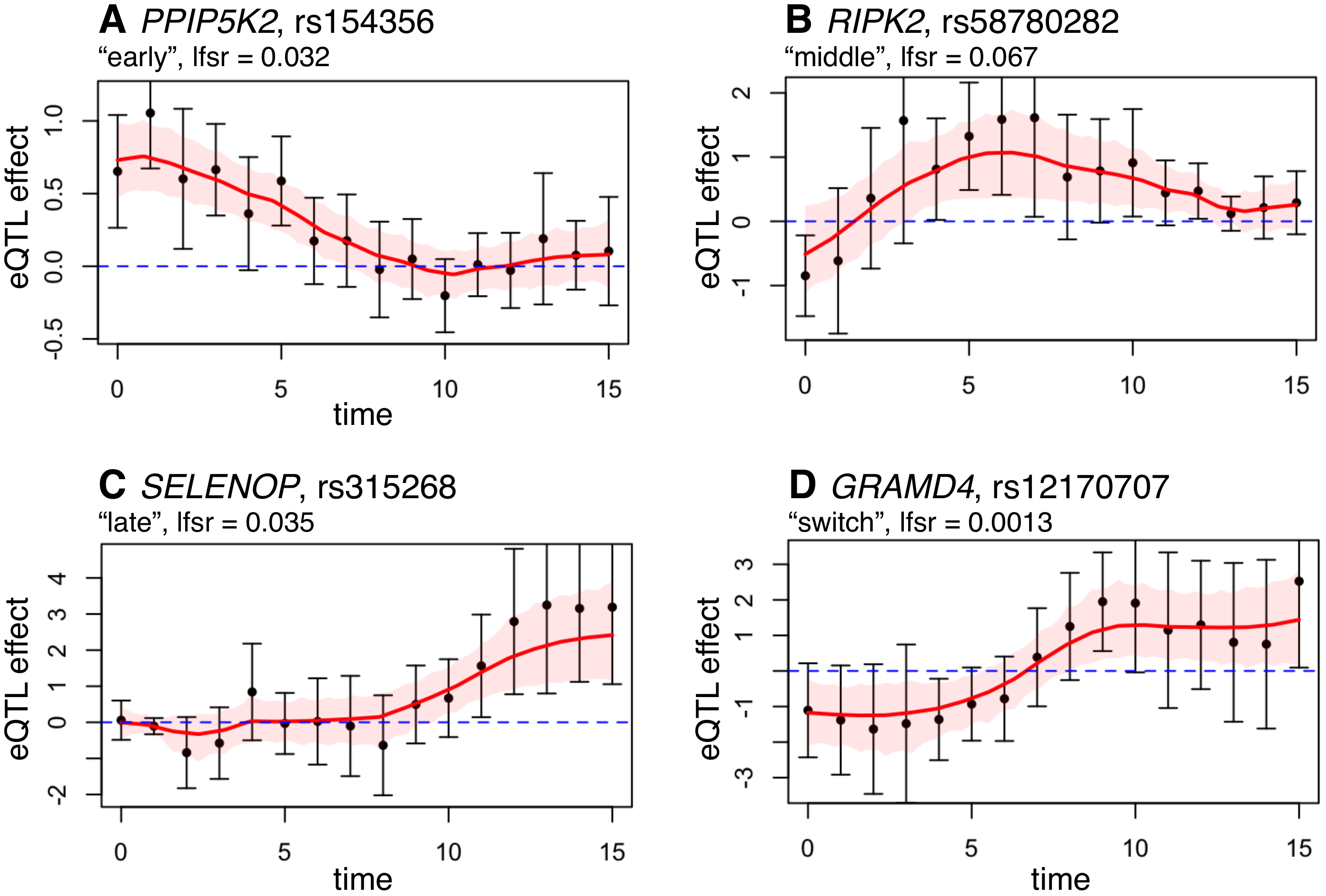}
\caption{
Examples of dynamic eQTLs categorized by FASH as being (A) {\em
early}, (B) {\em middle}, (C) {\em late} and (D) {\em
switch}. (See \cref{tab:eqtl_categories} for definitions of these
categories.)  The posterior mean of $\beta_j$ is shown as a red solid
line, and the 95\% credible interval is depicted by the shaded
region. Observed effect size estimates $\hat{\beta}_j$ are shown as
black dots, with vertical error bars representing $\pm 2$ the
(adjusted) standard errors $\tilde{s}_{jr}$. ``lfsr'' is the lfsr
obtained using the BF-adjusted prior.}
% 
% classified based on the false sign rate (FSR) of $5\%$.
% The labeling is the same as in \cref{fig:dynamic_highlight}.
%
% The horizontal blue dashed line is centered at zero.
%
\label{fig:examples_classification}
\end{figure}

\section{Discussion}
\label{discussion}

In this paper, we extended empirical Bayes (EB) ideas to reason about
posterior distributions on {\em functions.} This resulted in a
powerful and flexible modeling framework, FASH (``functional adaptive
shrinkage''), that can be used to test various hypotheses about
functions. FASH automatically adapts the priors by borrowing
information across all observation units, so the posterior inferences
should become more accurate as more data becomes available. In the
case study, where we used FASH to reanalyze dynamic eQTLs, a
particularly appealing aspect of FASH was being able categorize
the dynamic eQTLs into different types
(\cref{tab:eqtl_categories}).

% We illustrated these advantages through a re-analysis of the
% cardiomyocyte differentiation data of \citet{strober2019dynamic}.
% While motivated by dynamic eQTL detection, FASH more generally
% provides a useful tool whenever investigators wish to simultaneously
% estimate many unknown functions and test for deviations from baseline
% behavior.
%
% In the context of dynamic expression eQTL studies,
% our approach characterizes dynamic effect functions directly from
% summary statistics of each gene-variant pair, thereby improving
% portability and facilitating the integration of multiple gene-variant
% pairs.
% 

We also proposed a simple adjustment to the prior to address concerns
about miscalibration of the FDR and other posterior statistics when
the prior is misspecified. This adjustment results in more
conservative estimates of the prior---that is, greater weight on the
null proportion, $\pi_0$---and therefore it is more conservative in
its inferences. This proposal is not specific to FASH, and therefore
could potentially be used in other EB methods, particularly in
high-dimensional multivariate data settings \citep{mash, liu2024flexible}
where prior misspecification is likely to occur.

% \citep{soloff2024edge} As an EB-based approach, FASH estimates the
% prior distribution of effect functions semi-parametrically through
% mixtures of {GPs}.  As in related EB methods , the discrepancy
% between the true prior and its estimate can lead to calibration
% issues in {FDR} estimation \citep{soloff2024edge}.  To mitigate
% this, we outlined a simple yet effective {BF}-based adjustment that
% enforces conservative estimation of $\pi_0$ in the EB prior.  This
% adjustment is not specific to FASH, and could also be helpful in
% other EB frameworks, especially multivariate methods such as
% and \citet{liu2024flexible}, where prior misspecification may be of
% concern.

Although FASH defines posterior distributions on {\em functions}, the
actual computations involve probability distributions on
finite-dimensional spaces, making the computations tractable. To make
these tractable posterior computations scalable to large data sets, we
chose the ``$L$-GP'' family of priors because they produce posterior
computations with complexity that scales linearly in $J$, the number
of observation units, and $R_j$, the number of observations per unit;
more details on the complexity of the posterior computations are given
in Appendix~\ref{subsec:computation} of the supplement 
(The posterior computations with standard GPs scale
cubicly in $R_j$). Beyond computationally efficiency, the ``$L$-GP''
family also provides flexibility in hypothesis testing.
%
% In the current implementation, we restrict mixture components to the
% $L$-GP family, chosen for interpretability, suitability for
% hypothesis testing, and computational efficiency. In FASH, the
% computational complexity and memory cost scale linearly with the
% sample size $J$.  For the number of unique condition values $R_j$,
% standard GP computation typically incurs quadratic memory cost and
% cubic computational cost due to dense matrix operations.  Leveraging
% the Markovian structure of the $L$-GP, these costs are reduced to
% linear in $R_j$, which is particularly beneficial in studies where
% each variant is measured under many condition values .  }
%
In principle, FASH could be extended to other GP prior families,
motivated by other applications. For example, prior families based on
spatial Gaussian random fields \citep{lindgren2011explicit} could be of interest
applications involving spatial data.

% Another extension is to consider non-Gaussian likelihoods, such as
% Poisson models for correlated count data \citep{liu2024flexible},
% perhaps using Laplace approximation for tractable likelihood
% evaluation.  Accounting for approximation error when computing
% quantities such as lfdr may present some important challenges.  {In
% the dynamic eQTL studies presented in \cref{sec-application}, a
% potential concern is the presence of cell-line confounder effects,
% where subsets of cell lines exhibit distinct expression dynamics
% across many genes, potentially creating false positives in dynamic
% eQTL discoveries.  \cite{strober2019dynamic} empirically showed that
% controlling for cell line-collapsed principal components (PCs)
% effectively mitigates this issue.  Following the same principle, we
% also accounted for such confounder effects using the cell
% line-collapsed PCs in the analysis described
% in \cref{sec-application}.  Nevertheless, extending the FASH
% framework to explicitly model cell line-specific random effect
% functions over time would provide a more principled solution, albeit
% at the cost of increased model complexity and computational burden.

Similar to studying the genetic effects on gene expression over time
(``dynamic eQTLs''), there is also considerable interest in using
single-cell RNA-sequencing \citep{stegle-2015} to study how eQTL
effects vary along continuous cellular contexts such as cell
differentiation trajectories \citep{van2020single}. A FASH analysis of
such data could involve mapping the cells onto a 1-d trajectory, then
applying FASH to the effects estimated at the different points along
the trajectory.
%
% We view this as a promising extension and leave it for future work.
%

Finally, a practical point regarding FASH is that local significance
measures, such as lfsr and lfsr, are usually preferred over cumulative
significance measures such as FDR and FSR.
%
% we note that the default significance measures used in our examples
% are the cumulative FDR and FSR, to remain comparable to the original
% analysis \citep{strober2019dynamic}.  In general, we recommend that
% users threshold local significance measures such as the lfdr and
% lfsr.
%
We have found that when many of the observations have local measures
close to zero, this can have the effect of pulling down the cumulative
average over all the observation units, resulting in small FDR (or
FSR), even when the lfdr (or lfsr) is large.  In situations such as
this, using the local significance measures will reduce false
discoveries compared to the cumulative significance measures.

% To support its practical use, we have developed the
% open-source \texttt{fashr} software for the proposed method, available
% at \href{https://stephenslab.github.io/fashr/}{GitHub}.  We believe
% FASH provides a general and flexible EB framework for functional
% inference, broadly applicable to large-scale studies that aim to
% estimate many functional effects.

\section{Disclosure Statement}
\label{disclosure-statement}

The authors have no conflicts of interest to report.

\section{Data Availability Statement}
\label{data-availability-statement}

The summary statistics used in the dynamic eQTL analysis are provided in the online supplementary material, and the code to replicate the result is
available
on \href{https://stephenslab.github.io/fashr-paper/}{Github}.

\section{Acknowledgments}

The authors thank Kenneth Barr for his support with the cardiomyocyte data analyzed in \cref{sec-application}.

\section{Funding}

This research was supported in part by grants from the NSF
(DMS-2235451) and Simons Foundation (MPS-NITMB-00005320) to the
NSF-Simons National Institute for Theory and Mathematics in Biology
(NITMB).

\phantomsection\label{supplementary-material}
\bigskip

\begin{center}

% If you have supplementary material (e.g., software, data, technical
% proofs), identify them in the section below. In early stages of
% the submission process, you may be unsure what to include as
% supplementary material. Don't worry---this is something that can
% be worked out at later stages.
{\large\bf SUPPLEMENTARY MATERIAL}

\end{center}

\begin{description}
\item[Supplementary Text and Figures:]
Additional proofs and derivations, evaluation of FASH on simulated
data sets, FASH implementation details, and supplementary
figures. (PDF file)
\item[Supplementary Data:]
Summary statistics of eQTL effect estimates for all gene--variant pairs across all time points, as analyzed in \cref{sec-application}. (ZIP file)
\end{description}

% References with more than six authors may be abbreviated by listing
% the first six authors followed by "et al."
\setlength{\bibsep}{0.3ex plus 0.5ex}
\bibliography{fash}

\clearpage

\appendix

{\Large\bf Empirical Bayes Shrinkage of Functional Effects, \\[0.35ex]
with Application to Analysis of Dynamic eQTLs: \\[0.35ex]
Supplementary Text and Figures} \\[0.75ex]
{\large Ziang Zhang, Peter Carbonetto and Matthew Stephens}

\titlelabel{Appendix \thetitle: }

\section{Implementation Details}
\label{subsec:computation}

This section provides additional details on the implementation of FASH across its main computational steps.

\subsection*{Choice of the Grid for $\sigma_k$}

We fix $\sigma_0=0$ to represent the null component, and construct the remaining $K$ grid values $\{\sigma_k\}_{k=1}^K$ following the strategy in \citet{ash}. The goal is to create a sufficiently dense grid covering the interval $[\sigma_{\min}, \sigma_{\max}]$, such that the inferential results remain stable and do not change noticeably when using a larger or denser grid.

Conceptually, $\sigma_k$ controls the deviation from the baseline model $S_0$. Its interpretation in practice depends on the choice of operator $L$ defining the $L$-GP, as well as the measurement scale. To provide a more interpretable and consistent scale across applications, we consider an equivalent reparameterization in terms of the $h$-unit predictive standard deviation (PSD) \citep{zhang2024model, zhang2025efficient}, defined as
\begin{equation}\label{equ:psd}
    \sigma(h) = \text{SD}\!\left[\beta(t+h) \mid \beta(s): s \leq t \right], 
    \quad \beta \sim L\text{-GP}(\beta;\sigma).
\end{equation}
The PSD $\sigma(h)$ is a positive scaling of the original $\sigma$, but has a direct interpretation in terms of function variability that is comparable across different choices of $L$.

By default, we choose $[\sigma_{\min}, \sigma_{\max}]$ based on the one-unit PSD $\sigma(1)$, with grid values $\{\sigma_k\}_{k=1}^K$ equally spaced in $\text{log-precision}$ scale, i.e.\ $-2\log\sigma_k(1)$ between $0$ and $10$. While other settings can be explored, this default has been found to perform well in our experience.

\subsection*{Computing the Likelihood Matrix $\mathbf{L}$}

The main computational step of FASH is the evaluation of the likelihood matrix $\mathbf{L} \in \mathbb{R}^{J \times (K+1)}$. For each entry,
\begin{equation}\label{equ:marginal_lik}
\begin{aligned}
    \mathbf{L}_{jk} &= p_k(\boldsymbol{\hat{\beta}}_j \mid \boldsymbol{s}_j) \\
    &= \int p(\boldsymbol{\hat{\beta}}_j \mid \boldsymbol{\beta}_j, \boldsymbol{s}_j)\,
       p(\boldsymbol{\beta}_j \mid \sigma_k)\, d\boldsymbol{\beta}_j \\
    &= d\mathcal{N}(\boldsymbol{\hat{\beta}}_j; \mathbf{0}, \mathbf{C}_k + \mathbf{S}_j),
\end{aligned}
\end{equation}
where $\mathbf{C}_k$ is the covariance of $\boldsymbol{\beta}_j$ under $\beta_j \sim L\text{-GP}(\beta_j;\sigma_k)$ and $\mathbf{S}_j$ is a diagonal matrix of squared standard errors.

Evaluating \cref{equ:marginal_lik} requires the precision matrix $(\mathbf{C}_k+\mathbf{S}_j)^{-1}$. When $R_j$ is large, direct inversion becomes computationally expensive, with $O(R_j^3)$ time and $O(R_j^2)$ memory. 
However, unlike standard GPs, the $L$-GP has a Markovian structure due to its construction from differential operators \citep{sarkka2019applied}, which can be exploited to reduce computational complexity to $O(R_j)$.
This can be achieved either by augmenting $\boldsymbol{\beta}_j$ with its derivatives, or by finite element approximation \citep{lindgren2008second, zhang2024model, zhang2025efficient}. In this work, we adopt the latter approach using 20 equally spaced O-splines as the default.

\subsection*{Optimizing the Empirical Bayes Prior $\hat{g}_\beta$}

The empirical Bayes (EB) estimate $\hat{g}_\beta$ is obtained by maximizing the log-likelihood in \cref{eb_pi}. To encourage conservative estimation of $\pi_0$, a Dirichlet-based penalty can be added \citep{ash}:
\begin{equation}
    h(\boldsymbol{\pi};\lambda) = \prod_{k=0}^K \pi_k^{\lambda_k - 1}, \quad \lambda_k \geq 1,
\end{equation}
yielding the penalized log-likelihood
\begin{equation}\label{eb_pi2}
  l(\boldsymbol{\pi}) + \log h(\boldsymbol{\pi}) 
  = \sum_{j=1}^J \log \!\Big(\sum_{k=0}^K \pi_k \mathbf{L}_{jk}\Big) 
  + \sum_{k=0}^K (\lambda_k-1)\log \pi_k.
\end{equation}
This convex optimization problem can be solved efficiently with constrained algorithms; we use the sequential quadratic programming method of \citet{kim2020fast} with $\lambda_0=10$ and $\lambda_k=1$ for $k \neq 0$.

Although \cref{eb_pi2} assumes independence across {units}, in the presence of dependence it can be interpreted as a composite likelihood, and the EB estimate $\hat{\boldsymbol{\pi}}$ typically remains consistent \citep{mash}.

\subsection*{Posterior Computation}

Given $\mathbf{L}$ and $\hat{g}_\beta$, the posterior for each $\beta_j$ is obtained as the mixture in \cref{posterior_betaj}. Each component posterior $p_k(\beta_j \mid \boldsymbol{\hat{\beta}}, \boldsymbol{s})$ is a GP, so the overall posterior is a mixture of $K+1$ GPs. In practice, $\hat{\boldsymbol{\pi}}$ is typically sparse, so that only a few non-trivial mixture components remain, reducing the memory required to store posterior processes.

To compute lfsr in \cref{equ:lfsr} for a functional $\mathcal{F}$, we draw $M$ independent sample paths from the non-trivial components of the posterior mixture and approximate the relevant probabilities via Monte Carlo. The finite element representation of the $L$GP makes path simulation computationally efficient even for large $M$; we use $M=3000$ by default in our analysis.

\section{Additional proofs and derivations}
\label{sec:proofs}
\renewcommand{\thelemma}{S\arabic{lemma}}
\subsection*{Proof of Theorem 1}

\setcounter{theorem}{0}
\setcounter{lemma}{0}

\begin{theorem}[BF-adjustment gives conservative estimate]
Assume $\hat{\pi}_0$ is the adjusted estimate of $\pi_0=J_0/J$ obtained from Algorithm~\ref{algo:BF}, and that the $J_0$ null effects are i.i.d.\ from the null distribution specified in the adaptive prior. Then for any specification of the alternative distributions (components $1$ to $K$) in the prior and any buffer $\epsilon>0$, 
\[
\hat{\pi}_0 \;\ge\; \pi_0 \quad \text{almost surely as } J_0\to\infty.
\]
\end{theorem}

\begin{proof}
By Lemma~\ref{lemma:BF_moment}, under the null $\mathbb{E}_0(\mathrm{BF}_j)=1$ and $\mathrm{BF}_j>0$. 
Hence, by the strong law of large numbers,
\begin{equation}\label{eq:lln}
\lim_{J_0\to\infty}\frac{1}{J_0}\sum_{j\in\mathcal H_0}\mathrm{BF}_j \;=\; 1
\quad\text{almost surely}.
\end{equation}
From now on, all asymptotic arguments are understood almost surely, so we omit the notation.

Define the thresholded sets induced by $c^*$:
\[
\mathcal T_0=\{j:\mathrm{BF}_j<c^*\},\qquad 
\mathcal T_1=\{j:\mathrm{BF}_j\ge c^*\}.
\]
Let $J_{a,b}=|\mathcal H_a\cap \mathcal T_b|$ for $a,b\in\{0,1\}$, and set
\[
S=\sum_{j\in\mathcal H_0\cap\mathcal T_0}\mathrm{BF}_j.
\]
Then $J_{0,0}+J_{0,1}=J_0$, and
\begin{equation}\label{eq:pi_defs}
\hat{\pi}_0=\frac{J_{0,0}+J_{1,0}}{J}, 
\qquad 
\pi_0=\frac{J_{0,0}+J_{0,1}}{J}.
\end{equation}
Thus proving $\hat{\pi}_0\ge\pi_0$ is equivalent to showing $J_{1,0}\ge J_{0,1}$.

By \eqref{eq:lln}, there exists a positive integer $J_0'$ such that for all $J_0>J_0'$,
\begin{equation}\label{eq:null_constraint}
1+\frac{\epsilon}{2}\;\ge\;
\frac{\sum_{j\in\mathcal H_0}\mathrm{BF}_j}{J_0}
=\frac{\sum_{j\in\mathcal H_0\cap\mathcal T_1}\mathrm{BF}_j+S}{J_{0,1}+J_{0,0}}
\;\ge\;\frac{c^* J_{0,1}+S}{J_{0,1}+J_{0,0}},
\end{equation}
where the last inequality holds since $\mathrm{BF}_j\ge c^*$ on $\mathcal T_1$.

By the definition of $c^*$ in Algorithm~\ref{algo:BF}, we have $\mu(c^*)\ge 1+\epsilon$, hence
\begin{equation}\label{eq:mu_constraint}
1+\frac{\epsilon}{2}\;\le\;\mu(c^*) 
= \frac{\sum_{j\in\mathcal T_0}\mathrm{BF}_j}{J_{1,0}+J_{0,0}}
=\frac{\sum_{j\in\mathcal H_1\cap\mathcal T_0}\mathrm{BF}_j+S}{J_{1,0}+J_{0,0}}
\;\le\;\frac{c^* J_{1,0}+S}{J_{1,0}+J_{0,0}},
\end{equation}
where the last inequality holds since $\mathrm{BF}_j<c^*$ on $\mathcal T_0$.

Combining \eqref{eq:null_constraint} and \eqref{eq:mu_constraint}, for all large enough $J_0$,
\begin{equation}\label{eq:chain}
\frac{c^* J_{0,1}+S}{J_{0,1}+J_{0,0}}
\;\le\;1+\frac{\epsilon}{2}\;\le\;
\frac{c^* J_{1,0}+S}{J_{1,0}+J_{0,0}}.
\end{equation}

Define
\[
f(x)=\frac{c^*x+S}{x+J_{0,0}},\qquad 
f'(x)=\frac{c^*J_{0,0}-S}{(x+J_{0,0})^2}.
\]
Since $\mathrm{BF}_j<c^*$ on $\mathcal T_0$, we have $S<c^*J_{0,0}$. 
Therefore, $f'(x)>0$ for $x>0$, so $f$ is strictly increasing.  
From \eqref{eq:chain}, $f(J_{0,1})\le f(J_{1,0})$, which implies $J_{1,0}\ge J_{0,1}$.

Finally, by \eqref{eq:pi_defs}, this yields $\hat{\pi}_0\ge\pi_0$, completing the proof. \qedhere
\end{proof}

{\color{changed_color}
\subsection*{Conservativeness of posterior odds}
In this subsection, we give additional details showing that overestimating $\pi_0$ makes inference more conservative in favor of $H_0$, even when the alternative is misspecified.

For simplicity, assume $\tilde{\pi}_0$ is fixed with $0<\pi_0\le \tilde{\pi}_0<1$, and denote
\[
p_0 = p_0(\hat{\boldsymbol{\beta}}_j\mid \boldsymbol{s}_j),\quad
p_1 = p_1(\hat{\boldsymbol{\beta}}_j\mid \boldsymbol{s}_j),\quad
\tilde p_1 = \tilde p_1(\hat{\boldsymbol{\beta}}_j\mid \boldsymbol{s}_j),
\]
as the null, true alternative, and misspecified alternative marginal densities, respectively.

A useful measure of the evidence against $H_0$ is the \emph{posterior odds} (PO), defined by
\begin{equation}\label{eq:PO_def}
  \mathrm{PO}
  = \frac{p(H_1\mid \hat{\boldsymbol{\beta}}_j,\boldsymbol{s}_j)}
         {p(H_0\mid \hat{\boldsymbol{\beta}}_j,\boldsymbol{s}_j)}
  = \frac{1-\pi_0}{\pi_0}\,\frac{p_1}{p_0}.
\end{equation}
Replacing $\pi_0$ and $p_1$ by their adjusted/misspecified counterparts yields the fitted $\widetilde{\mathrm{PO}}$:
\begin{equation}\label{eq:PO_tilde_def}
  \widetilde{\mathrm{PO}}
  = \frac{1-\tilde\pi_0}{\tilde\pi_0}\,\frac{\tilde p_1}{p_0}.
\end{equation}

First, it is straightforward to show, under the null, the fitted $\widetilde{\mathrm{PO}}$ is more conservative than the true $\mathrm{PO}$ in expectation:

\begin{lemma}[Conservative under $H_0$]\label{lemmaS1}
If $\tilde\pi_0 \ge \pi_0$, then
\[
\mathbb{E}_{H_0}[\widetilde{\mathrm{PO}}] \;\le\; \mathbb{E}_{H_0}[\mathrm{PO}].
\]
\end{lemma}

\begin{proof}
From \eqref{eq:PO_def}-\eqref{eq:PO_tilde_def} and taking expectation with respect to $p_0$,
\[
\mathbb{E}_{H_0}[\widetilde{\mathrm{PO}}]
= \frac{1-\tilde\pi_0}{\tilde\pi_0}\,
  \mathbb{E}_{p_0}\!\left[\frac{\tilde p_1}{p_0}\right]
= \frac{1-\tilde\pi_0}{\tilde\pi_0},
\quad
\mathbb{E}_{H_0}[\mathrm{PO}]
= \frac{1-\pi_0}{\pi_0}\,
  \mathbb{E}_{p_0}\!\left[\frac{p_1}{p_0}\right]
= \frac{1-\pi_0}{\pi_0},
\]
since $\mathbb{E}_{H_0}\!\left[\frac{r(\hat{\boldsymbol{\beta}})}{p_0(\hat{\boldsymbol{\beta}})}\right] = \int \frac{r(\hat{\boldsymbol{\beta}})}{p_0(\hat{\boldsymbol{\beta}})}p_0(\hat{\boldsymbol{\beta}})\,d\hat{\boldsymbol{\beta}} =\int r(\hat{\boldsymbol{\beta}})\,d\hat{\boldsymbol{\beta}}=1$ for any density $r$.
Because $x\mapsto \frac{1-x}{x}$ is strictly decreasing on $(0,1)$ and $\tilde\pi_0\ge \pi_0$, the claim follows.
\end{proof}

Under the alternative, conservativeness is seen via an analogous argument applied to log-PO.

\begin{lemma}[Conservative under $H_1$]\label{lemmaS2}
If $\tilde\pi_0 \ge \pi_0$, then
\[
\mathbb{E}_{H_1}\!\big[\log \widetilde{\mathrm{PO}}\big]
\;\le\;
\mathbb{E}_{H_1}\!\big[\log \mathrm{PO}\big].
\]
\end{lemma}

\begin{proof}
Under $H_1$,
\[
\begin{aligned}
\mathbb{E}_{H_1}[\log \widetilde{\mathrm{PO}}]
&= \log\frac{1-\tilde{\pi}_0}{\tilde{\pi}_0}
 + \mathbb{E}_{p_1}\!\left[\log \frac{\tilde p_1}{p_0}\right] \\
&= \log\frac{1-\tilde{\pi}_0}{\tilde{\pi}_0}
 + \mathbb{E}_{p_1}\!\left[\log \frac{p_1}{p_0}\right]
 - \mathbb{E}_{p_1}\!\left[\log \frac{p_1}{\tilde p_1}\right] \\
&= \log\frac{1-\tilde{\pi}_0}{\tilde{\pi}_0}
 + D_{\mathrm{KL}}(p_1\|p_0) - D_{\mathrm{KL}}(p_1\|\tilde p_1) \\
&\le \log\frac{1-\tilde{\pi}_0}{\tilde{\pi}_0} + D_{\mathrm{KL}}(p_1\|p_0)
\;\le\; \log\frac{1-\pi_0}{\pi_0} + D_{\mathrm{KL}}(p_1\|p_0)
= \mathbb{E}_{H_1}[\log \mathrm{PO}],
\end{aligned}
\]
since $D_{\mathrm{KL}}(\cdot\|\cdot)\ge 0$ and $x\mapsto \log\frac{1-x}{x}$ is strictly decreasing on $(0,1)$ with $\tilde\pi_0\ge \pi_0$.
\end{proof}

}

\renewcommand{\thefigure}{S\arabic{figure}}
\renewcommand{\theHfigure}{S\arabic{figure}} 
\crefname{figure}{Fig.}{Figs.}     % \cref -> “Fig. S1”
\Crefname{figure}{Figure}{Figures} % \Cref -> “Figure S1”
\setcounter{figure}{0}

\section{Simulation Study}
\label{sec:simulations}

We conducted a simulation study to evaluate the performance of FASH with the proposed BF-based adjustment for estimating the null proportion $\pi_0$. 
The two hypothesis settings were testing whether the effect function is constant or linear, consistent with \cref{sec-application}.

We generated $J=1000$ independent {observation units}, each with an effect function $\beta_j$ drawn from one of the following three categories:
\begin{enumerate}[label=\Alph*.]
    \item \textit{Non-dynamic:} $\beta_j = c_j$ with $c_j \sim N(0,1)$.
    \item \textit{Linear-dynamic:} $\beta_j = c_j + b_j t$ with $c_j \sim N(0,1)$ and $b_j \sim N(0,1/4)$.
    \item \textit{Nonlinear-dynamic:} $\beta_j$ sampled from a standard $\text{IWP}_2(\sigma)$ with $\sigma(16)=5$.
\end{enumerate}

Each function was observed at sixteen equally spaced time points $t=0,1,\dots,16$. Standard errors $s_{j,r}$ were drawn independently from $\{0.1,0.3,0.5\}$ with equal probability. Examples of simulated {observations} are shown in \cref{fig:appendixB_example}.
\begin{figure}[!t]
    \centering
    \includegraphics[width=1\textwidth]{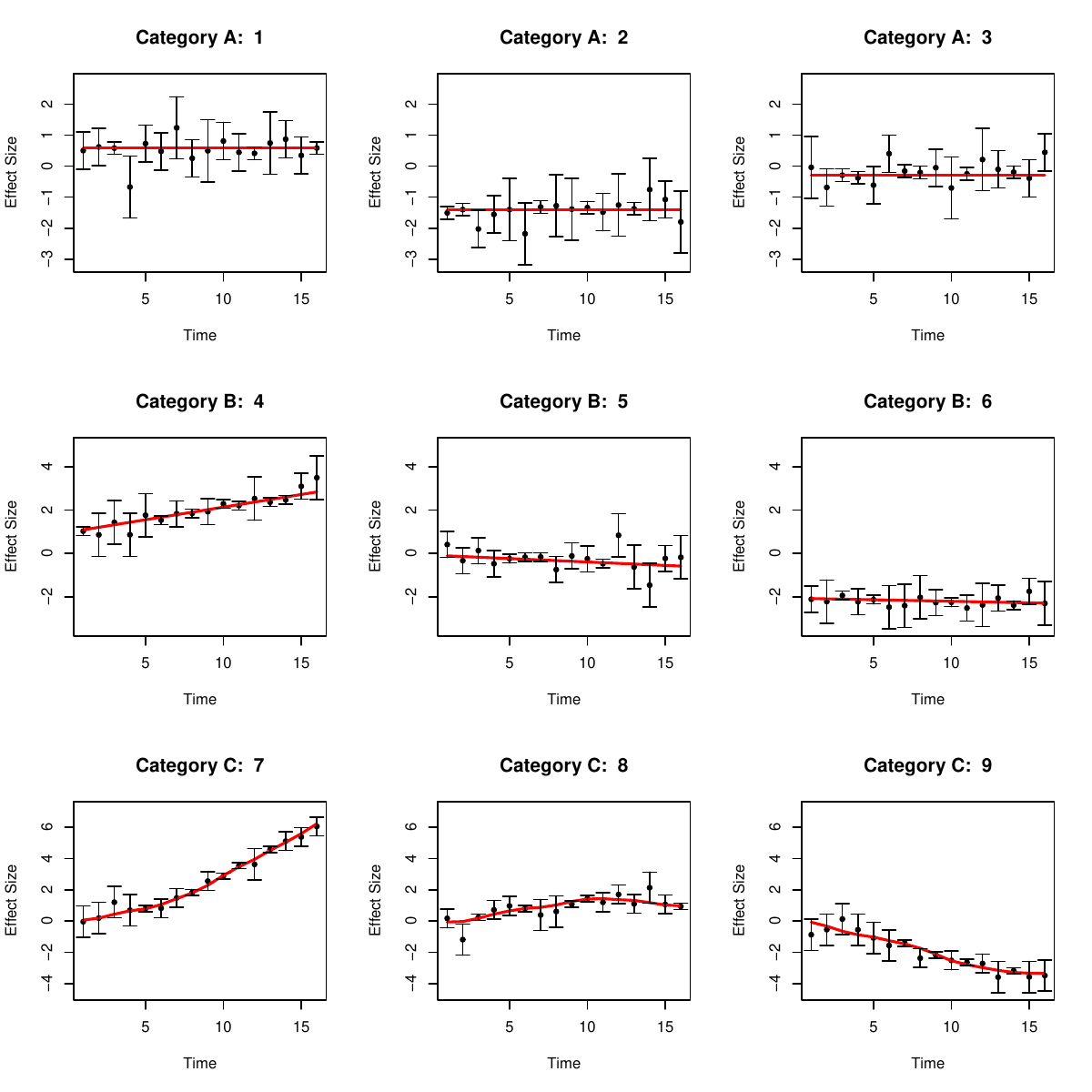}
    \caption{Examples of simulated {observations}: points indicate observed effect estimates; vertical bars denote $\pm 2$ standard errors; red lines denote the true effect function $\beta_j$.}
    \label{fig:appendixB_example}
\end{figure}
To vary the underlying null proportion, we introduced a parameter $\rho \in [0.05,0.5]$ with increment $0.01$. For each $\rho$, we simulated $J(1-\rho)$ {observations} from category A, $J\rho/2$ from category B, and the remainder from category C. Thus, when testing $S_0$ as constant functions, the true null proportion is $\pi_0=1-\rho$, and when testing $S_0$ as linear functions, the true null proportion is $\pi_0=1-\rho/2$.

We fit the FASH model as in \cref{sec-application}, and compared the unpenalized MLE $\hat{\pi}_0$ with its BF-adjusted counterpart. Results are summarized in \cref{fig:appendixB_result}. 
Panels (a-b) demonstrate that the BF-adjusted estimates consistently remain above the true $\pi_0$ across replications, confirming the conservative property stated in \cref{thrm:BF_control}. In contrast, the raw MLE often underestimates $\pi_0$, particularly when testing $S_0$ as constant functions. Panels (c-d) further show that using a denser grid for $\sigma$ does not affect this conclusion: the BF-adjusted estimates still maintain the desired conservativeness.

\begin{figure}[!t]
    \centering
    \includegraphics[width=1\textwidth]{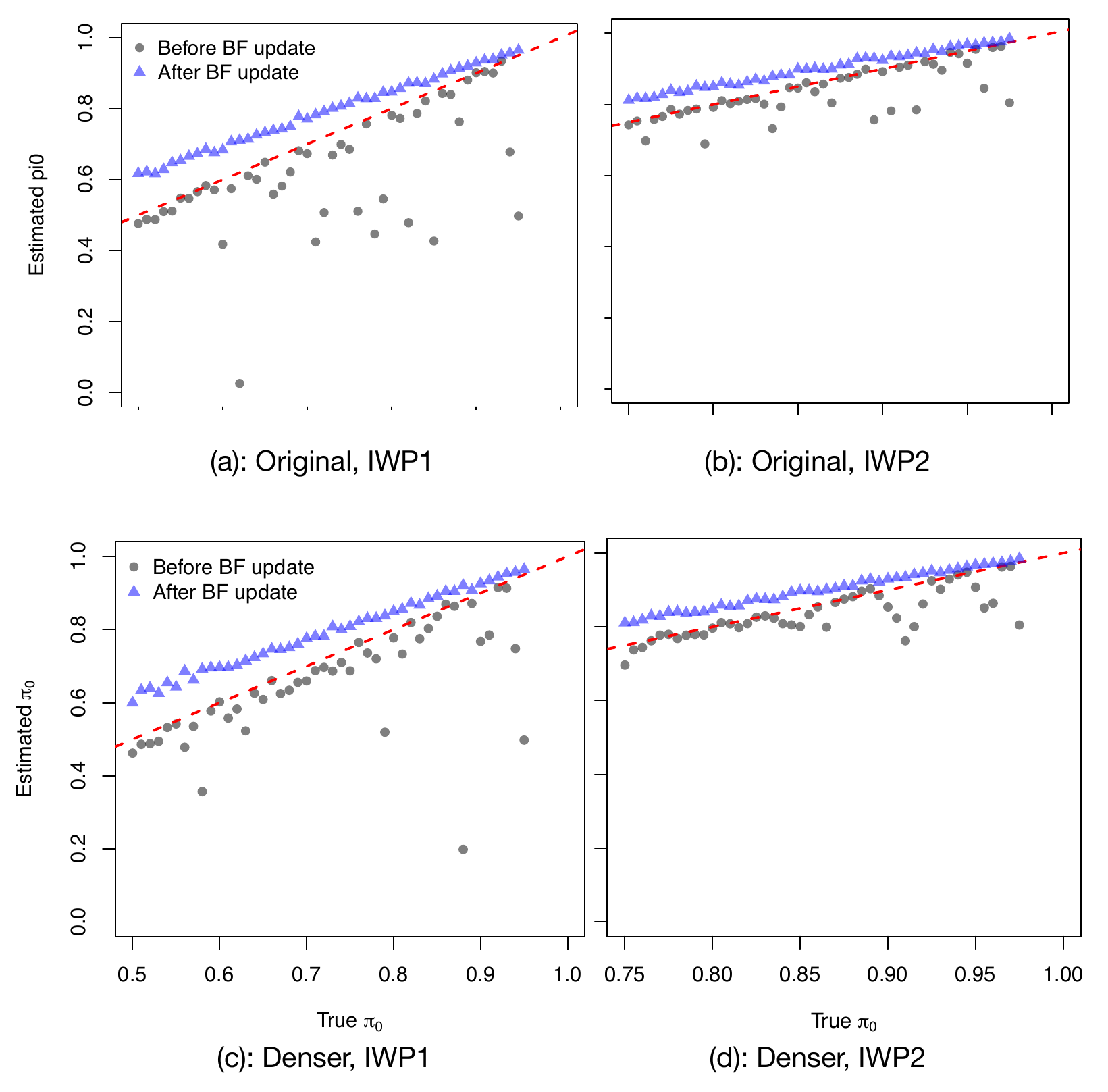}
    \caption{Estimated null proportions $\hat{\pi}_0$ in the simulation study. The left column corresponds to testing constant functions ($S_0=$ constant, fitted with $\text{IWP}_1$); the right column corresponds to testing linear functions ($S_0=$ linear, fitted with $\text{IWP}_2$). Black: MLE; blue: BF-adjusted estimates; red: 45-degree line. 
    (a-b) results using the original grid in \cref{sec-application} (spacing $0.2$ on log-precision scale). 
    (c-d) results using a denser grid (spacing $0.1$).}
    \label{fig:appendixB_result}
\end{figure}

Next, we examine a specific setting with $\rho=0.2$, which corresponds to a true null proportion of $\pi_0=0.8$ when testing for dynamic eQTLs and $\pi_0=0.9$ when testing for nonlinear dynamic eQTLs. We apply FASH as in \cref{sec-application} to address two inferential goals: (i) detecting dynamic eQTLs and (ii) detecting nonlinear dynamic eQTLs. For a range of nominal FDR levels $\alpha$, we evaluate the empirical FDR of FASH with and without BF-adjustment. As shown in \cref{fig:appendixB_result2}, when $\alpha \leq 0.05$, the empirical FDR is already well controlled without adjustment. However, as $\alpha$ increases, the unadjusted results exhibit inflated FDR, particularly when testing for nonlinear dynamic eQTLs. In contrast, the BF-adjusted version consistently controls FDR at or below the nominal level, in agreement with the theoretical guarantee of \cref{thrm:BF_control}. 

In terms of power, the BF-adjustment yields slightly more conservative results, leading to a modest reduction in power. This loss, however, is small—capped at about 5\% in the worst case—and the BF-adjusted FASH still achieves over 80\% power when $\alpha=0.05$ for both tasks. Therefore, we view this as a reasonable trade-off to ensure FDR remains below the nominal level, and recommend applying the BF-adjustment in practice. 
Examples of significant discoveries at $\alpha=0.05$ are shown in \cref{fig:appendixB_result3}, with the posterior effect functions obtained from FASH compared against the true underlying effects.

\begin{figure}[!t]
    \centering
\includegraphics[width=1\textwidth]{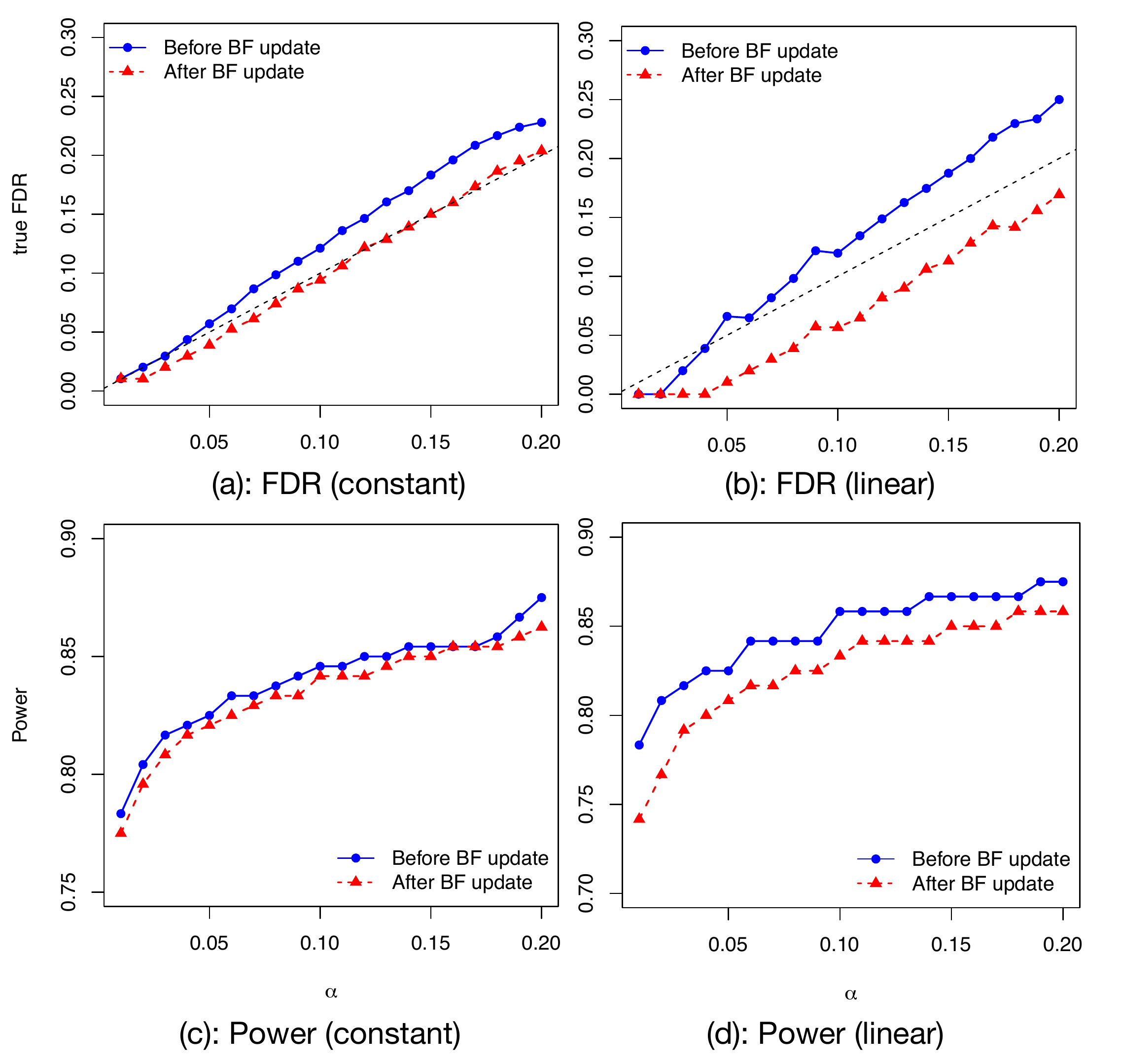}
    \caption{
    Empirical FDR (a-b) and power (c-d) of FASH under the two inferential goals, plotted against the nominal FDR $\alpha$, with (red) or without (blue) the BF-adjustment of $\pi_0$. 
    The left column corresponds to testing dynamic eQTLs ($H_0:\beta_j$ constant), and the right column to testing nonlinear dynamics ($H_0:\beta_j$ linear). 
    The black dashed line in (a-b) is the 45-degree reference line for FDR calibration.
    }
    \label{fig:appendixB_result2}
\end{figure}

\begin{figure}[!t]
    \centering
\includegraphics[width=1\textwidth]{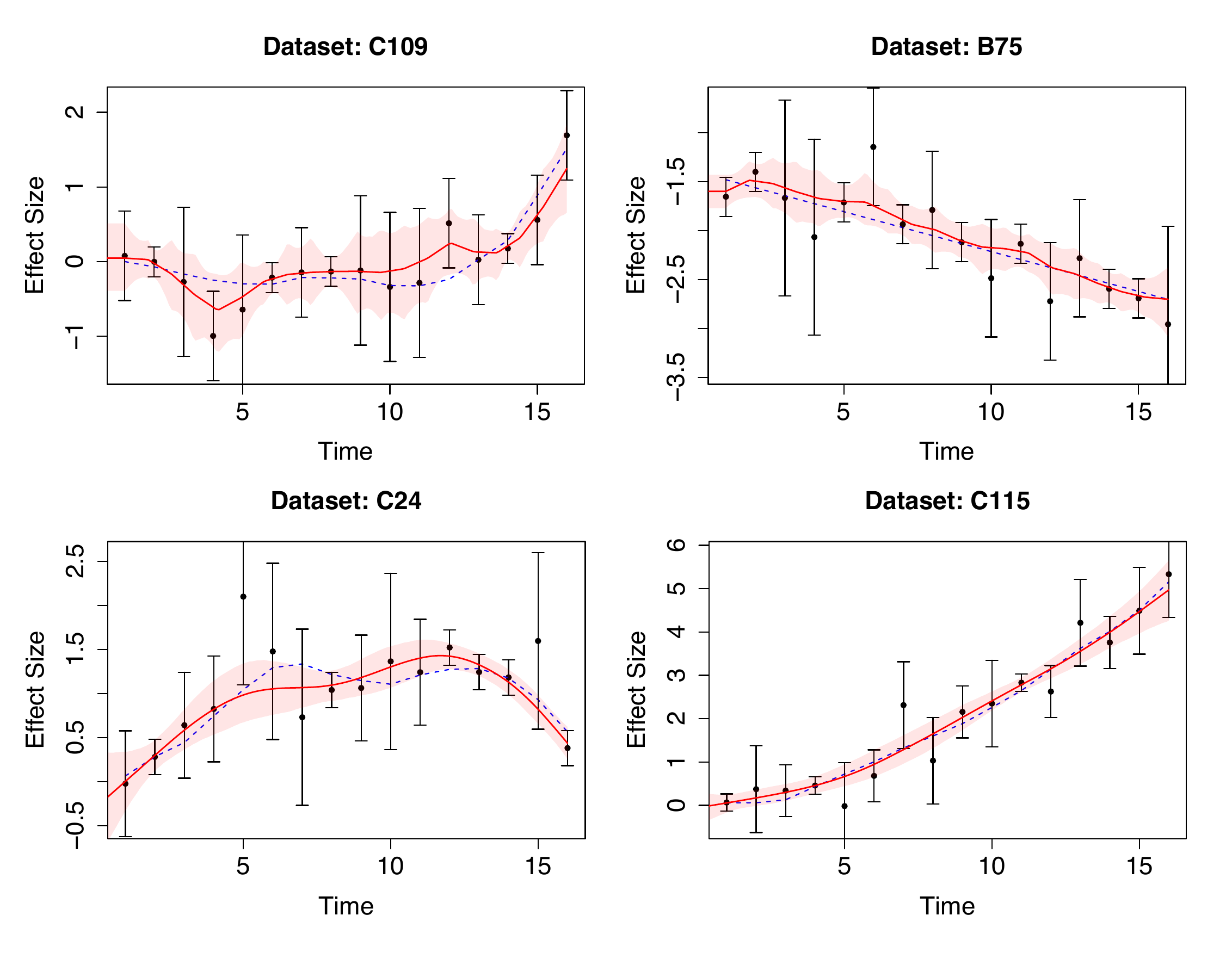}
    \caption{
    Examples of significant eQTLs identified by FASH in the simulation study at a nominal FDR level of $\alpha=0.05$ (top: dynamic eQTLs; bottom: nonlinear< dynamic eQTLs). 
    The posterior mean of \( \beta_j \) is shown as a red solid line, with the 95\% credible interval shaded in red. 
    Observed effect size estimates \( \hat{\beta}_j \) are shown as black dots with error bars denoting \( \pm 2\tilde{s}_{j,r} \). 
    The true underlying effect function is shown as a blue dashed line.
    }
    \label{fig:appendixB_result3}
\end{figure}

\section{Additional Details of the iPSC Cardiomyocyte Differentiation Study}
\label{subsec:application-details}

For each gene-variant pair $j \in [J]$, we obtained eQTL effect
estimates $\hat{\beta}_j(t_r)$ at 16 time points by fitting linear
regression models, separately at each time point $r \in [16]$. We used
the same linear regression model that was used
in \cite{strober2019dynamic}:
% 
% At each of the $R = 16$ time points $t_r$, $r \in [16]$, we estimated
% the eQTL effect size for the $j$th gene-variant pair using the linear
% regression model that was used in
%
\begin{equation}
\label{equ:eqtl}
E_{jc}(t_r) \mid {\bm z}_{cr}, G_c
\overset{\text{ind}}{\sim} N({\bm z}_{cr}^T\boldsymbol{b}_{jr} + 
G_c\beta_j(t_r), \sigma_{E_{jr}}^2), 
\quad r \in [16], c \in [C_r], j \in [J],
\end{equation}
in which $E_{jc}(t_r) \in \mathbb{R}$ denotes the standardized
expression of gene $j$ in cell line $c$ at time point $t_r$,
$G_c \in \{0, 1, 2\}$ is the genotype dosage of the SNP in cell line
$c$, ${\bm z}_{cr} \in \mathbb{R}^6$ is the vector of covariates
(intercept and first 5 PCs), and $\beta_j(t_r) \in \mathbb{R}$, ${\bm
b}_{jr} \in \mathbb{R}^6$ are the regression coefficients to be
estimated. Principal components (PCs) were computed from the
``cell-line-collapsed'' expression matrix of dimension $19 \times
\mbox{212,147}$ \citep{strober2019dynamic}. The number of cell lines observed
at time $t_r$, denoted by $C_r$, was 19 for most time points, except
for day~14 (18 cell lines) and days 3 and 5 (16 cell lines). The
residual variance $\sigma_{E_{jr}}^2$ at time point $t_r$ was assumed
to be constant across cell lines at a given time point. From
%
% time-point specific regression 
%
\eqref{equ:eqtl}, we computed the $t$
statistic for each gene-variant pair on each day as $T_{jr}
= \hat{\beta}_j(t_r) / s_{jr}$, in which $s_{jr}$ is the standard error
of $\hat{\beta}_j(t_r)$.

Under the null hypothesis that $\beta_j(t_r) = 0$, $T_{jr}$ follows a
$t$ distribution with $\nu_r = C_r - 5$ degrees of freedom.  If the
number of cell lines measured at each time point were large, the
asymptotic normality of the maximum-likelihood estimator would imply
that $T_{jr}$ approximately follows a standard normal distribution,
and hence $\hat{\beta}_j(t_r) \sim N(\beta_j(t_r),
s_{jr}^2)$. However, since only 19 cell lines were included in this
study, this approximation is not reliable. In fact, inflation of the
type I error rate was reported under this setting (see Fig. S10 in the
supplement of \citealt{strober2019dynamic}). To account inflation of
type I error s in a scalable way,
following \citep{lu2019empiricalbayesestimationnormal} we defined
``$t$-adjusted standard errors'' $\tilde{s}_{jr}$ as satisfying
\begin{equation}
\label{equ:adj-SE}
\Phi(\hat{\beta}_j(t_r) / \tilde{s}_{jr}) = 
P_{T_{\nu_r}}(\hat{\beta}_j(t_r) / s_{jr}),
\end{equation}
where $\Phi$ denotes the CDF of the standard normal distribution and
$P_{T_{\nu_r}}$ is the CDF of the $t$ distribution with $\nu_r$
degrees of freedom.  The effect estimates $\hat{\beta}_{jr}$ and the
adjusted standard errors $\tilde{s}_{jr}$ were then the inputs to
FASH.

The priors \eqref{equ:fash_adaptive} used in the the FASH analyses
were all defined on an equally spaced grid on the log-scale, with
$\sigma_0 = 0, \sigma_1 = e^{-5} \approx 0.01$, $\sigma_{52} = 1$,
and $\log \sigma_{k+1} = \log \sigma_k + 0.1$, $k = 2, \ldots, 51$.

\section{Supplementary Figures \& Tables}
\label{subsec:sup_figure}

% Gene-set enrichment results for all dynamic + switch (no filtering)
\footnotesize
\setlength{\tabcolsep}{3pt}
\renewcommand{\arraystretch}{1.08}
\begin{longtable}{l p{0.46\linewidth} c c >{\ttfamily}r >{\ttfamily}r}
\caption{Gene-set enrichment results for highlighted genes with dynamic eQTLs (Dynamic) and for genes with switch dynamic eQTLs (Switch).
All gene sets are from the MSigDB Hallmark (HALLMARK\_*) collection; the prefix is omitted in the table.}
\label{tab:hallmark_enrichment_all_vs_switch}\\

\toprule
Category & Gene set & GeneRatio & BgRatio & $p$-value & $q$-value \\
\midrule
\endfirsthead

\toprule
Category & Gene set & GeneRatio & BgRatio & $p$-value & $q$-value \\
\midrule
\endhead

\bottomrule
\endfoot

% -------------------- All dynamic eQTLs --------------------
Dynamic & HYPOXIA                         & 25/1177 & 89/6362  & 0.0169 & 0.436 \\
Dynamic & IL6\_JAK\_STAT3\_SIGNALING       & 8/1177  & 21/6362  & 0.0280 & 0.436 \\
Dynamic & ESTROGEN\_RESPONSE\_EARLY        & 22/1177 & 82/6362  & 0.0394 & 0.436 \\
Dynamic & ESTROGEN\_RESPONSE\_LATE         & 21/1177 & 79/6362  & 0.0476 & 0.436 \\
Dynamic & ANDROGEN\_RESPONSE               & 16/1177 & 57/6362  & 0.0499 & 0.436 \\
Dynamic & HEME\_METABOLISM                 & 22/1177 & 85/6362  & 0.0564 & 0.436 \\
Dynamic & NOTCH\_SIGNALING                 & 5/1177  & 15/6362  & 0.1276 & 0.750 \\
Dynamic & KRAS\_SIGNALING\_DN              & 8/1177  & 28/6362  & 0.1306 & 0.750 \\
Dynamic & GLYCOLYSIS                       & 23/1177 & 100/6362 & 0.1498 & 0.750 \\
Dynamic & MYOGENESIS                       & 16/1177 & 67/6362  & 0.1624 & 0.750 \\
Dynamic & UV\_RESPONSE\_DN                 & 16/1177 & 68/6362  & 0.1780 & 0.750 \\
Dynamic & KRAS\_SIGNALING\_UP              & 11/1177 & 47/6362  & 0.2415 & 0.840 \\
Dynamic & TNFA\_SIGNALING\_VIA\_NFKB        & 13/1177 & 58/6362  & 0.2665 & 0.840 \\
Dynamic & HEDGEHOG\_SIGNALING              & 4/1177  & 15/6362  & 0.2962 & 0.840 \\
Dynamic & APOPTOSIS                        & 15/1177 & 70/6362  & 0.3072 & 0.840 \\
Dynamic & INTERFERON\_GAMMA\_RESPONSE      & 12/1177 & 56/6362  & 0.3359 & 0.840 \\
Dynamic & MITOTIC\_SPINDLE                 & 28/1177 & 139/6362 & 0.3399 & 0.840 \\
Dynamic & WNT\_BETA\_CATENIN\_SIGNALING    & 5/1177  & 21/6362  & 0.3455 & 0.840 \\
Dynamic & COAGULATION                      & 7/1177  & 31/6362  & 0.3459 & 0.840 \\
Dynamic & P53\_PATHWAY                     & 17/1177 & 83/6362  & 0.3626 & 0.840 \\
Dynamic & ADIPOGENESIS                     & 20/1177 & 106/6362 & 0.5008 & 0.945 \\
Dynamic & \begin{tabular}[t]{@{}l@{}}REACTIVE\_OXYGEN\_SPECIES\_\\PATHWAY\end{tabular}
                                            & 5/1177  & 26/6362  & 0.5410 & 0.945 \\
Dynamic & APICAL\_SURFACE                  & 3/1177  & 15/6362  & 0.5439 & 0.945 \\
Dynamic & PI3K\_AKT\_MTOR\_SIGNALING        & 10/1177 & 55/6362  & 0.5794 & 0.945 \\
Dynamic & COMPLEMENT                       & 10/1177 & 57/6362  & 0.6278 & 0.945 \\
Dynamic & IL2\_STAT5\_SIGNALING            & 13/1177 & 75/6362  & 0.6499 & 0.945 \\
Dynamic & XENOBIOTIC\_METABOLISM           & 12/1177 & 71/6362  & 0.6837 & 0.945 \\
Dynamic & INTERFERON\_ALPHA\_RESPONSE      & 4/1177  & 27/6362  & 0.7633 & 0.945 \\
Dynamic & APICAL\_JUNCTION                 & 13/1177 & 82/6362  & 0.7739 & 0.945 \\
Dynamic & DNA\_REPAIR                      & 13/1177 & 82/6362  & 0.7739 & 0.945 \\
Dynamic & ANGIOGENESIS                     & 2/1177  & 15/6362  & 0.7956 & 0.945 \\
Dynamic & EPITHELIAL\_MESENCHYMAL\_TRANSITION
                                            & 13/1177 & 84/6362  & 0.8031 & 0.945 \\
Dynamic & PROTEIN\_SECRETION               & 9/1177  & 61/6362  & 0.8208 & 0.945 \\
Dynamic & ALLOGRAFT\_REJECTION             & 5/1177  & 36/6362  & 0.8225 & 0.945 \\
Dynamic & TGF\_BETA\_SIGNALING             & 5/1177  & 36/6362  & 0.8225 & 0.945 \\
Dynamic & UNFOLDED\_PROTEIN\_RESPONSE      & 11/1177 & 75/6362  & 0.8442 & 0.945 \\
Dynamic & FATTY\_ACID\_METABOLISM          & 10/1177 & 72/6362  & 0.8812 & 0.945 \\
Dynamic & MYC\_TARGETS\_V2                 & 5/1177  & 41/6362  & 0.8992 & 0.945 \\
Dynamic & BILE\_ACID\_METABOLISM           & 3/1177  & 28/6362  & 0.9132 & 0.945 \\
Dynamic & SPERMATOGENESIS                  & 3/1177  & 28/6362  & 0.9132 & 0.945 \\
Dynamic & UV\_RESPONSE\_UP                 & 10/1177 & 76/6362  & 0.9174 & 0.945 \\
Dynamic & INFLAMMATORY\_RESPONSE           & 5/1177  & 43/6362  & 0.9206 & 0.945 \\
Dynamic & OXIDATIVE\_PHOSPHORYLATION       & 16/1177 & 119/6362 & 0.9445 & 0.945 \\
Dynamic & PEROXISOME                       & 4/1177  & 44/6362  & 0.9739 & 0.945 \\
Dynamic & G2M\_CHECKPOINT                  & 16/1177 & 136/6362 & 0.9881 & 0.945 \\
Dynamic & MTORC1\_SIGNALING                & 13/1177 & 129/6362 & 0.9973 & 0.945 \\
Dynamic & CHOLESTEROL\_HOMEOSTASIS         & 2/1177  & 42/6362  & 0.9981 & 0.945 \\
Dynamic & E2F\_TARGETS                     & 10/1177 & 127/6362 & 0.9998 & 0.945 \\
Dynamic & MYC\_TARGETS\_V1                 & 6/1177  & 124/6362 & 1.0000 & 0.945 \\

\midrule
% -------------------- Switch dynamic eQTLs --------------------
Switch & HYPOXIA                         & 11/250 & 89/6362  & 0.000668 & 0.0246 \\
Switch & KRAS\_SIGNALING\_UP             & 7/250  & 47/6362  & 0.00218  & 0.0348 \\
Switch & MYOGENESIS                      & 8/250  & 67/6362  & 0.00445  & 0.0348 \\
Switch & P53\_PATHWAY                    & 9/250  & 83/6362  & 0.00501  & 0.0348 \\
Switch & GLYCOLYSIS                      & 10/250 & 100/6362 & 0.00565  & 0.0348 \\
Switch & ANDROGEN\_RESPONSE              & 7/250  & 57/6362  & 0.00656  & 0.0348 \\
Switch & COAGULATION                     & 5/250  & 31/6362  & 0.00662  & 0.0348 \\
Switch & IL6\_JAK\_STAT3\_SIGNALING       & 4/250  & 21/6362  & 0.00822  & 0.0378 \\
Switch & NOTCH\_SIGNALING                & 3/250  & 15/6362  & 0.0192   & 0.0787 \\
Switch & APICAL\_JUNCTION                & 7/250  & 82/6362  & 0.0415   & 0.153 \\
Switch & INTERFERON\_GAMMA\_RESPONSE     & 5/250  & 56/6362  & 0.0678   & 0.227 \\
Switch & \begin{tabular}[t]{@{}l@{}}REACTIVE\_OXYGEN\_SPECIES\_\\PATHWAY\end{tabular}
      & 3/250 & 26/6362 & 0.0802 & 0.246 \\
Switch & ESTROGEN\_RESPONSE\_LATE         & 6/250  & 79/6362  & 0.0894   & 0.253 \\
Switch & ESTROGEN\_RESPONSE\_EARLY        & 6/250  & 82/6362  & 0.1025   & 0.268 \\
Switch & EPITHELIAL\_MESENCHYMAL\_TRANSITION
                                      & 6/250  & 84/6362  & 0.1117   & 0.268 \\
Switch & HEME\_METABOLISM                & 6/250  & 85/6362  & 0.1165   & 0.268 \\
Switch & APOPTOSIS                       & 5/250  & 70/6362  & 0.1398   & 0.299 \\
Switch & XENOBIOTIC\_METABOLISM          & 5/250  & 71/6362  & 0.1459   & 0.299 \\
Switch & PI3K\_AKT\_MTOR\_SIGNALING       & 4/250  & 55/6362  & 0.1690   & 0.328 \\
Switch & TNFA\_SIGNALING\_VIA\_NFKB       & 4/250  & 58/6362  & 0.1927   & 0.355 \\
Switch & MYC\_TARGETS\_V2                & 3/250  & 41/6362  & 0.2171   & 0.381 \\
Switch & INFLAMMATORY\_RESPONSE          & 3/250  & 43/6362  & 0.2380   & 0.399 \\
Switch & UV\_RESPONSE\_DN                & 4/250  & 68/6362  & 0.2778   & 0.445 \\
Switch & KRAS\_SIGNALING\_DN             & 2/250  & 28/6362  & 0.3018   & 0.463 \\
Switch & UNFOLDED\_PROTEIN\_RESPONSE     & 4/250  & 75/6362  & 0.3405   & 0.502 \\
Switch & COMPLEMENT                      & 3/250  & 57/6362  & 0.3895   & 0.529 \\
Switch & MTORC1\_SIGNALING               & 6/250  & 129/6362 & 0.3964   & 0.529 \\
Switch & ALLOGRAFT\_REJECTION            & 2/250  & 36/6362  & 0.4164   & 0.529 \\
Switch & TGF\_BETA\_SIGNALING            & 2/250  & 36/6362  & 0.4164   & 0.529 \\
Switch & HEDGEHOG\_SIGNALING             & 1/250  & 15/6362  & 0.4523   & 0.555 \\
Switch & MITOTIC\_SPINDLE                & 6/250  & 139/6362 & 0.4670   & 0.555 \\
Switch & OXIDATIVE\_PHOSPHORYLATION      & 5/250  & 119/6362 & 0.5046   & 0.581 \\
Switch & MYC\_TARGETS\_V1                & 5/250  & 124/6362 & 0.5415   & 0.600 \\
Switch & WNT\_BETA\_CATENIN\_SIGNALING   & 1/250  & 21/6362  & 0.5697   & 0.600 \\
Switch & IL2\_STAT5\_SIGNALING           & 3/250  & 75/6362  & 0.5705   & 0.600 \\
Switch & ADIPOGENESIS                    & 4/250  & 106/6362 & 0.6043   & 0.618 \\
Switch & G2M\_CHECKPOINT                 & 5/250  & 136/6362 & 0.6244   & 0.622 \\
Switch & INTERFERON\_ALPHA\_RESPONSE     & 1/250  & 27/6362  & 0.6620   & 0.642 \\
Switch & FATTY\_ACID\_METABOLISM         & 2/250  & 72/6362  & 0.7817   & 0.736 \\
Switch & CHOLESTEROL\_HOMEOSTASIS        & 1/250  & 42/6362  & 0.8153   & 0.736 \\
Switch & PEROXISOME                      & 1/250  & 44/6362  & 0.8297   & 0.736 \\
Switch & DNA\_REPAIR                     & 2/250  & 82/6362  & 0.8392   & 0.736 \\
Switch & PROTEIN\_SECRETION              & 1/250  & 61/6362  & 0.9143   & 0.783 \\
Switch & UV\_RESPONSE\_UP                & 1/250  & 76/6362  & 0.9534   & 0.798 \\
Switch & E2F\_TARGETS                    & 1/250  & 127/6362 & 0.9942   & 0.814 \\

\end{longtable}
\normalsize

\begin{figure}[!t]
\centering
\includegraphics[width=\textwidth]{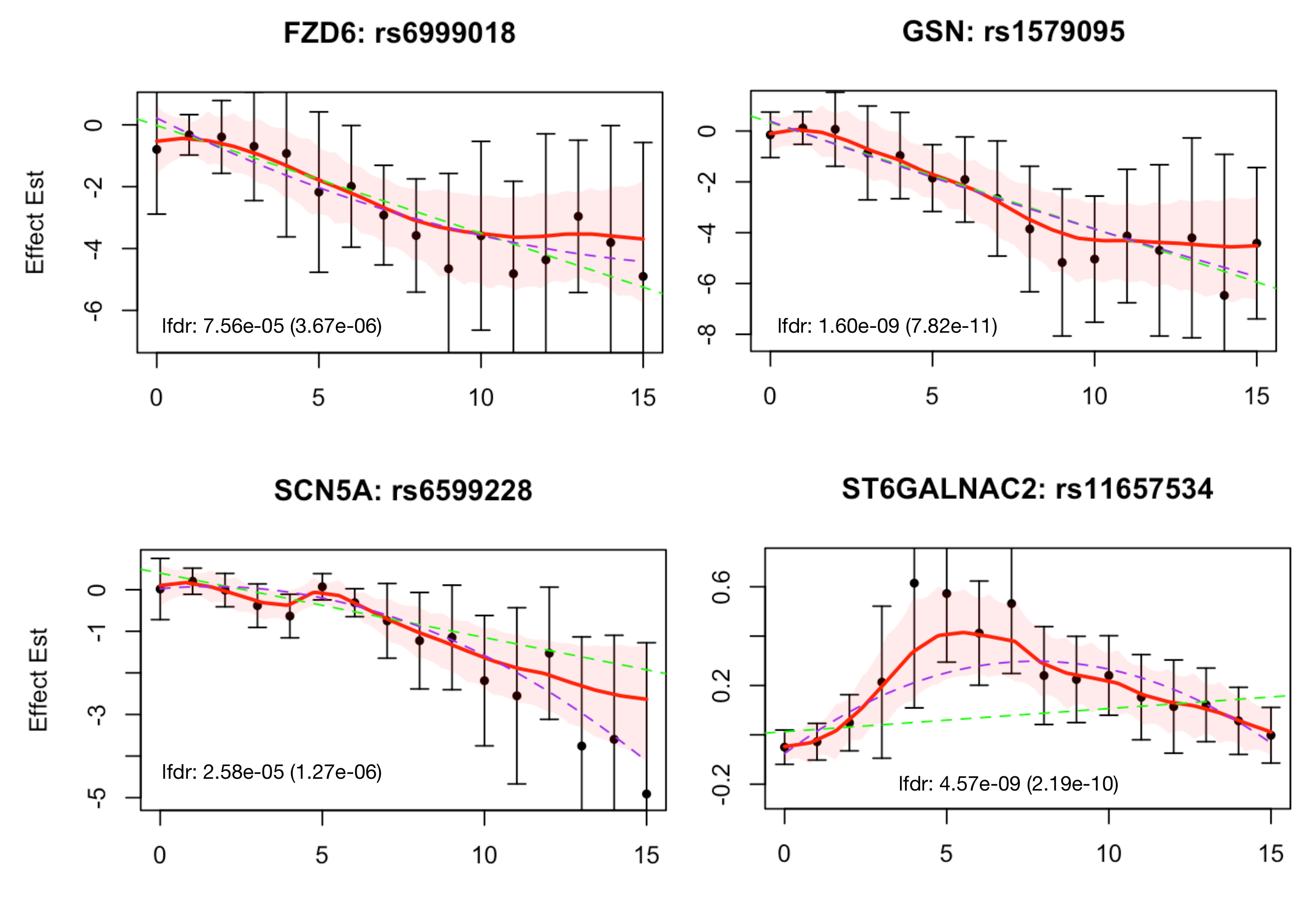}
\caption{Selected gene-variant pairs with significant dynamic eQTL effects 
{based on a 5\% FDR threshold}. In each plot, the
posterior mean of $\beta_j$ is shown as a red solid line and the 95\%
credible interval is shown by the red shaded region. Observed effect
size estimates $\hat{\beta}_j$ are shown as black dots and vertical
bars represent $\pm 2$ the (adjusted) standard errors
$\tilde{s}_{jr}$. In ``lfdr: x (y)'', x is the lfdr after the BF
adjustment, and y is the lfdr before the  BF adjustment. For comparison,
the inverse-variance weighted least squares estimates of
$\hat{\beta}_j$ are shown as green (linear) and purple (quadratic)
dashed lines. Variants rs6999018 and rs1579095 appear to have
minimal impact on the expression of their respective
genes, \textit{FZD6} and \textit{GSN}, during the first three days of
differentiation, but begin to significantly downregulate expression at
later stages.  {These genes have been previously implicated in
cardiomyocyte development and function.}  For
example, \cite{mazzotta2016distinctive} reported that \textit{FZD6}
plays an important regulatory role in cardiomyocyte differentiation
through non-canonical Wnt signaling, while \cite{li2009gelsolin}
showed that \textit{GSN} contributes to cardiac remodeling via DNase
I-mediated apoptosis.}
\label{fig:dynamic_highlight}
\end{figure}

\begin{figure}[!h]
    \centering
    \includegraphics[width=1\textwidth]{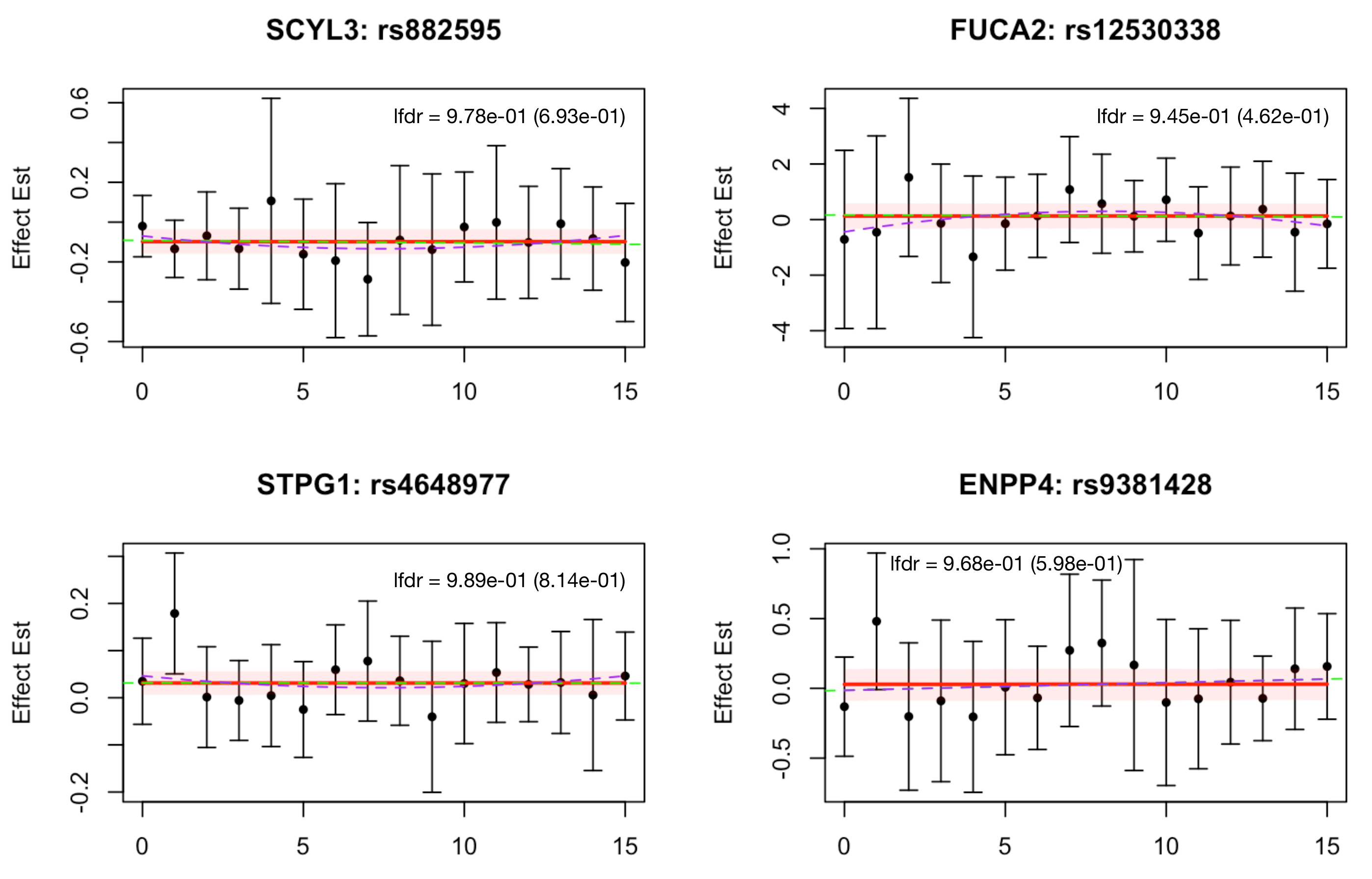}
    \caption{Selected variant-gene pairs showing no significant dynamic eQTL effects {based on the $5\%$ false discovery rate threshold}. The posterior mean of \( \beta_j \) is shown as a red solid line, with the 95\% credible interval indicated by the red shaded region. Observed effect size estimates \( \hat{\beta}_j \) are shown as black dots, with vertical bars representing \( \pm 2\tilde{s}_{j,r} \). 
    {The lfdr for each pair is shown, with the value before BF adjustment given in parentheses.
    For comparison, inverse-variance weighted least squares estimates of $\hat{\beta}_j$ are shown as a green dashed line (linear) and a purple dashed line (quadratic).}
    }
    \label{fig:nonsig_order1}
\end{figure}

\begin{figure}[!t]
    \centering
    \includegraphics[width=1\textwidth]{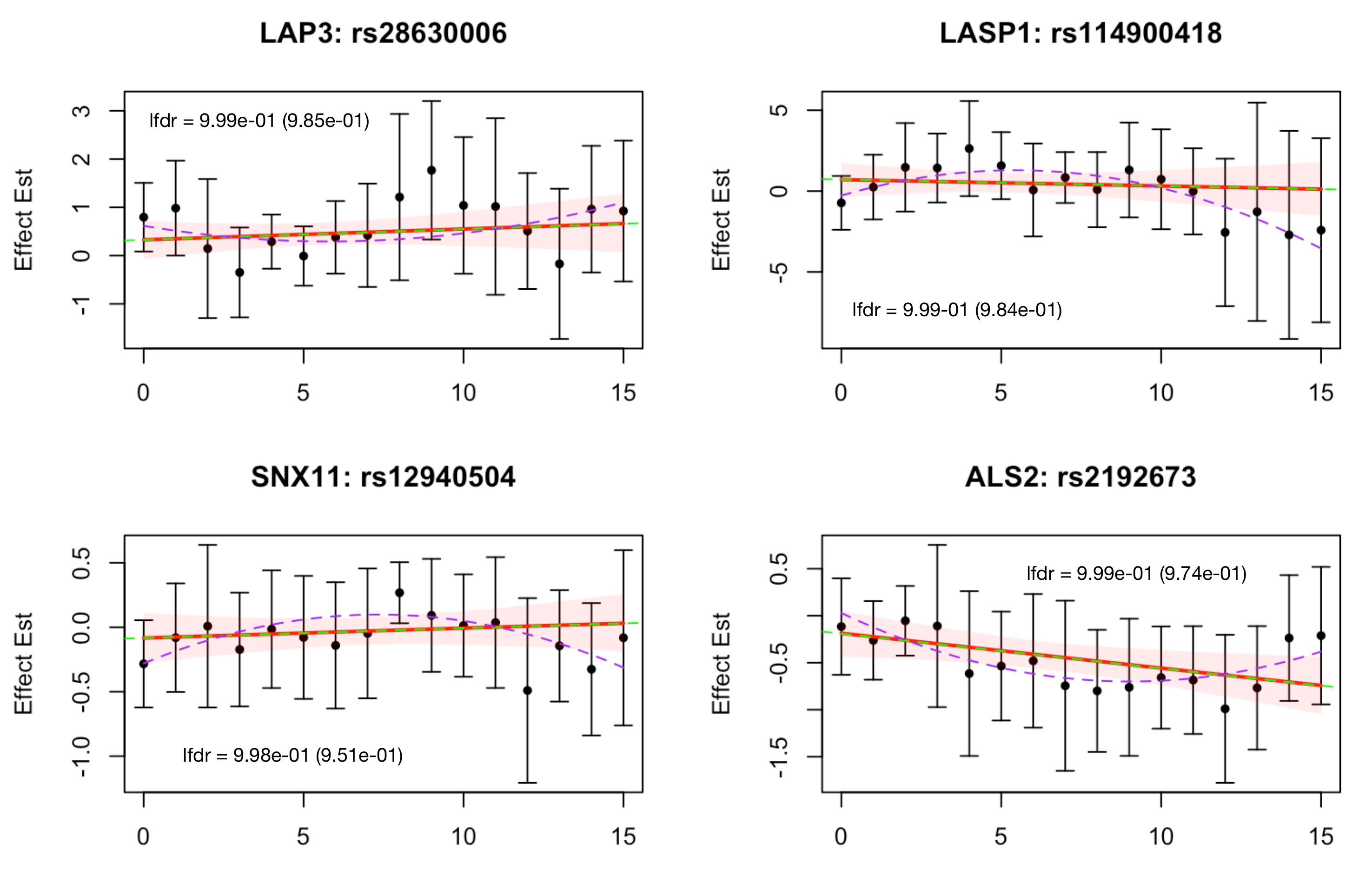}
    \caption{{Selected variant-gene pairs showing no significant nonlinear dynamic eQTL effects {based on the $5\%$ false discovery rate threshold}. The posterior mean of \( \beta_j \) is shown as a red solid line, with the 95\% credible interval indicated by the red shaded region. Observed effect size estimates \( \hat{\beta}_j \) are shown as black dots, with vertical bars representing \( \pm 2\tilde{s}_{j,r} \). 
    {The lfdr for each pair is shown, with the value before BF adjustment given in parentheses.
    For comparison, inverse-variance weighted least squares estimates of $\hat{\beta}_j$ are shown as a green dashed line (linear) and a purple dashed line (quadratic).}
    }}
    \label{fig:nonsig_order2}
\end{figure}

\begin{figure}[!t]
    \centering
    \includegraphics[width=1\textwidth]{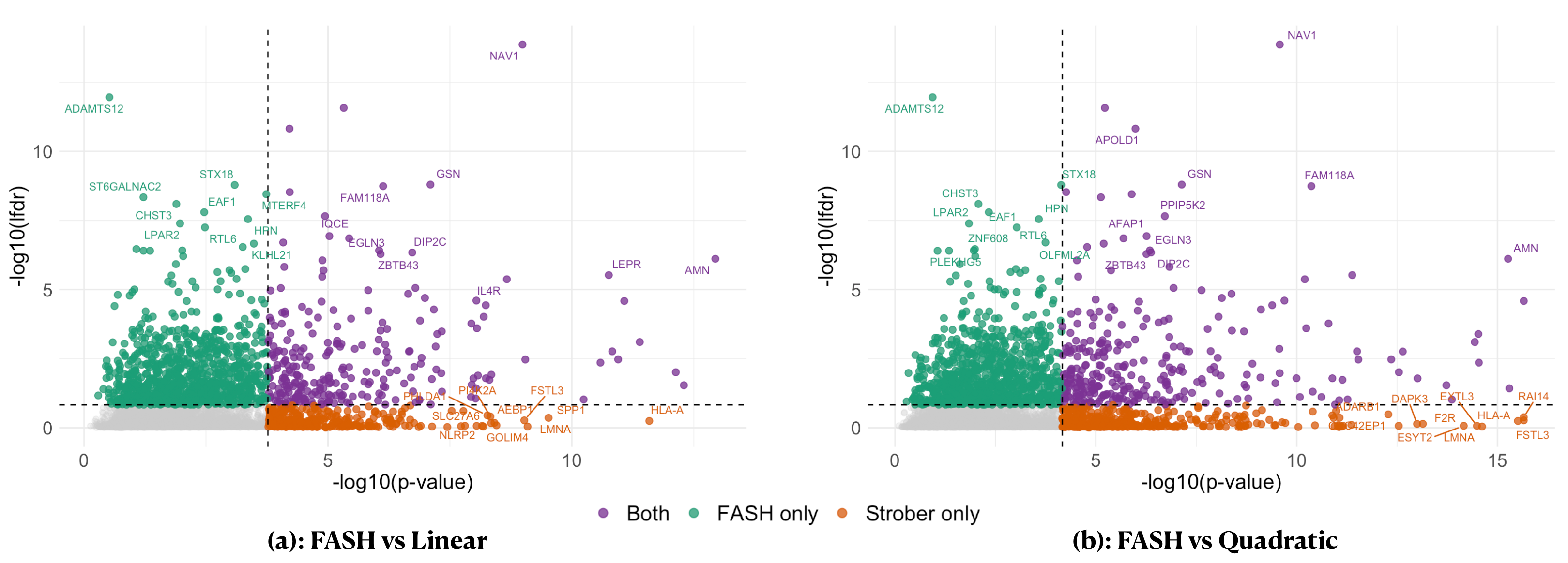}
    \caption{
    Scatterplot of ($-\log_{10}$) lfdr from FASH-$\mathrm{IWP}_1$ versus ($-\log_{10}$) p-values from \cite{strober2019dynamic}, based on the linear (a) or quadratic interaction (b) test.
    Each point represents a unique gene; for each gene, the lfdr and p-value are taken from the variant-gene pair with the smallest lfdr and the smallest p-value, respectively.
    The horizontal dashed line indicates the lfdr cutoff corresponding to FDR $=0.05$.
    The colors indicate significance relative to the FDR threshold (grey: significant in neither; green: significant only by lfdr; orange: significant only by p-value; purple: significant in both).
    }
    \label{fig:lfdr_pval}
\end{figure}

\begin{figure}[!t]
    \centering
    \includegraphics[width=1\textwidth]{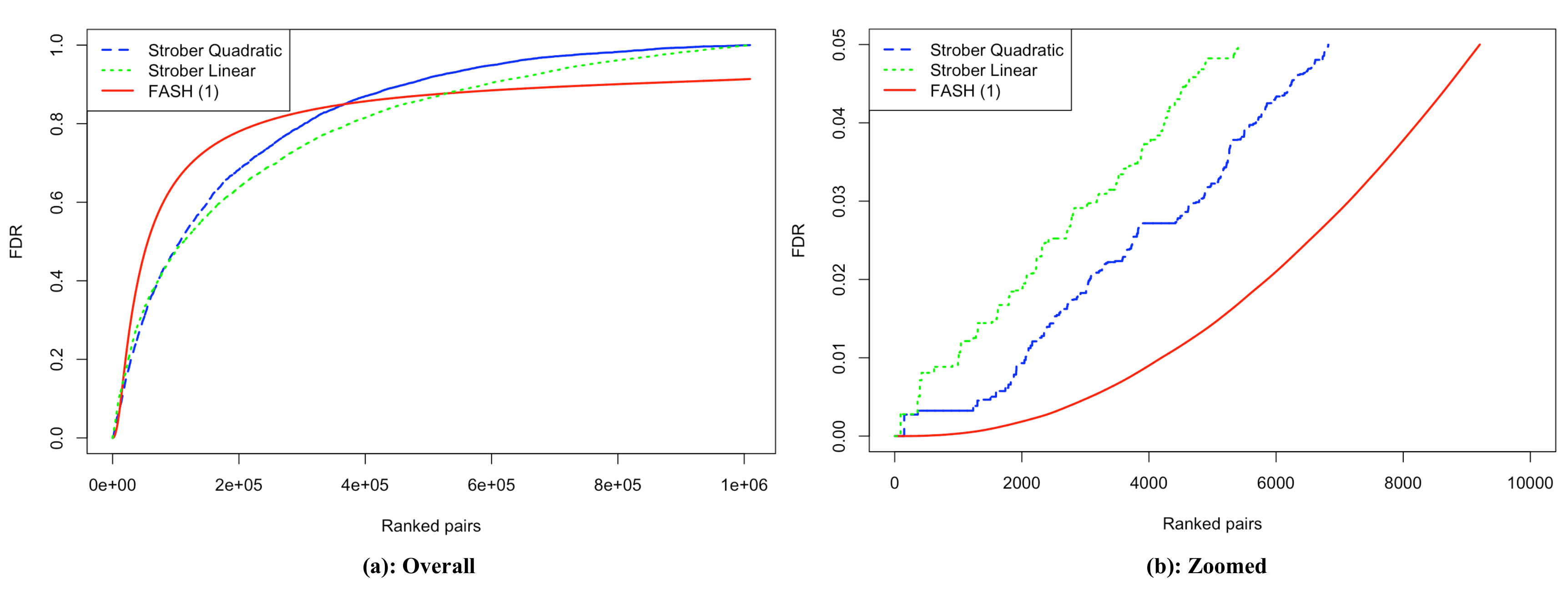}
    \caption{
    False discovery rate (FDR) along the ranked list of variant-gene pairs (from most to least significant), computed from the local false discovery rates (lfdrs) produced by FASH-$\mathrm{IWP}_1$ (solid red) and by applying the eFDR procedure of \cite{gamazon2013integrative} to the p-values reported in \cite{strober2019dynamic} (linear test: dotted green; quadratic test: dashed blue). 
    }
    \label{fig:FDR_eFDR}
\end{figure}

\begin{figure}[!t]
    \centering
    \includegraphics[width=1\textwidth]{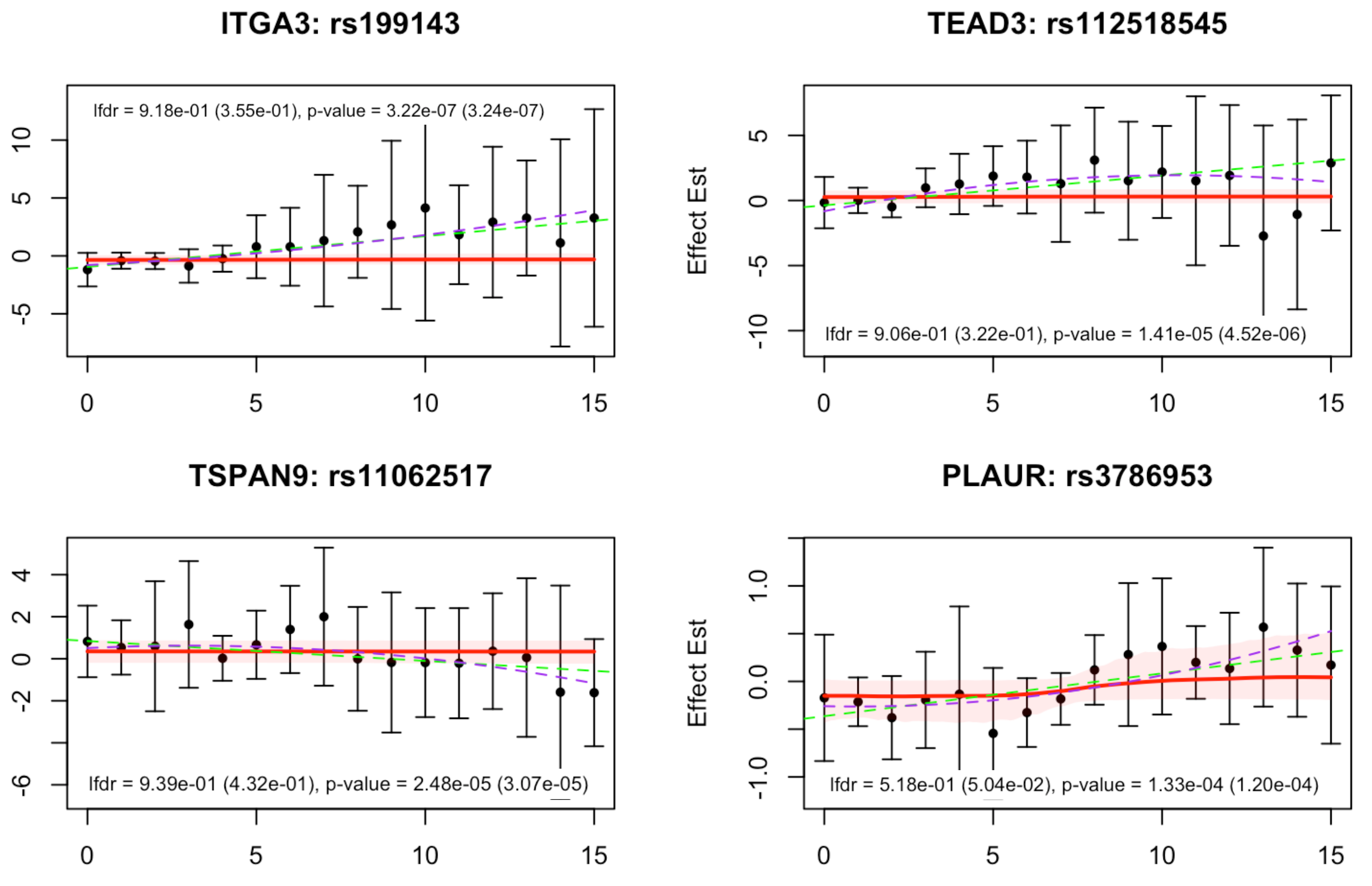}
    \caption{
    Examples of genes identified in \cite{strober2019dynamic} but not detected by FASH-$\mathrm{IWP}_1$.
    For each variant-gene pair, the posterior mean of $\beta_j$ is shown as a solid red line, with the 95\% credible interval indicated by the red shaded band.
    Observed effect size estimates $\hat{\beta}_j$ are shown as black dots, with vertical bars representing $\pm 2\tilde{s}_{j,r}$.
    The lfdr for each pair is reported, with the pre--BF-adjustment value shown in parentheses.
    The (linear/quadratic) p-values from \cite{strober2019dynamic} are also reported, with the quadratic p-value shown in parentheses.
    }
    \label{fig:missFASH}
\end{figure}

\end{document}